\documentclass[journal]{IEEEtran}
\usepackage{amsmath,amsfonts}
\usepackage{amssymb,,amsthm}
\usepackage{algorithmic}
\usepackage{algorithm}
\usepackage{array}
\newcolumntype{C}[1]{>{\centering\arraybackslash}p{#1}}
\usepackage[caption=false,font=normalsize,labelfont=sf,textfont=sf]{subfig}
\usepackage{textcomp}
\usepackage{pifont}
\usepackage{stfloats}
\usepackage{url}
\usepackage{verbatim}
\usepackage{graphicx}
\usepackage{xcolor}
\usepackage{booktabs}
\usepackage{multirow}
\usepackage{adjustbox}
\usepackage{bm}
\usepackage{hyperref} 
\newtheorem{lemma}{Lemma}
\theoremstyle{definition}
\newtheorem{definition}{Definition}
\newtheorem{theorem}{Theorem}
\usepackage[dvipsnames]{xcolor}
\usepackage{mathtools}
\usepackage{makecell}
\usepackage{tikz}

\DeclareRobustCommand*\circled[1]{%
  \tikz[baseline=(char.base)]{%
    \node[shape=circle,draw,inner sep=1pt] (char) {#1};}%
}

\usepackage[backend=biber,
style=ieee,
sorting=none,
doi=false,isbn=false,url=false ]{biblatex}

\usepackage[colorinlistoftodos]{todonotes}

\usepackage[framemethod=tikz]{mdframed}
\newmdenv[
  topline=true,
  bottomline=true,
  leftline=false,
  rightline=false,
  linewidth=0.6pt,
  innerleftmargin=0pt,
  innerrightmargin=0pt,
  innertopmargin=4pt,
  innerbottommargin=4pt,
]{modelrule}

\usepackage[normalem]{ulem}
\usepackage{xcolor}

\begin{document}


\title{Structure-preserving Optimal Kron-based Reduction of Radial Distribution Networks}


\author{Omid Mokhtari,~\IEEEmembership{Member,~IEEE}, 
        Samuel Chevalier,~\IEEEmembership{Member,~IEEE}, 
        and Mads Almassalkhi,~\IEEEmembership{Senior Member,~IEEE}
        
\thanks{This material is based upon work supported by the U.S. Department of Energy’s Office of Energy Efficiency and Renewable Energy (EERE) under the Enabling Place-Based Renewable Power Generation using Community Energyshed Design initiative, award number DE-EE0010407. The views expressed herein do not necessarily represent the views of the U.S. Department of Energy or the United States Government.\textbf{}

O.~Mokhtari, S.~Chevalier, and M. Almassalkhi are with the Department of Electrical and Biomedical Engineering, University of Vermont, Burlington, VT 05405, USA
{\{omid.mokhtari, schevali, malmassa\}@uvm.edu}.}
}

\maketitle

\begin{abstract}
 Network reduction simplifies complex electrical networks to address computational challenges of large-scale transmission and distribution grids. Traditional network reduction methods are often based on a predefined set of nodes or lines to remain in the reduced network.
 This paper builds upon previous work on optimal Kron-based reduction of networks, which was formulated as a mixed-integer linear program, to enhance the framework in three aspects. First, the scalability is improved via a cutting plane restriction, tightened Big~M bounds, and a zero-injection node reduction stage.
 Next, we introduce a radiality-preservation step to identify and recover nodes whose restoration ensures radiality.
 A linearized voltage magnitude error constraint is incorporated to explicitly bound the difference between full and reduced networks.
 The model is validated through its application to the 533-bus distribution test system and a 3499-bus utility feeder for a set of representative loading scenarios. In the 533-bus system, an 85\% reduction was achieved with a maximum voltage error below 0.0025 per unit, while in the 3499-bus feeder, over 94\% reduction was obtained with maximum voltage errors below 0.002 per unit. Additionally, we show that the radialization step accelerates the runtime of optimal voltage control problems when applied to Kron-reduced networks.
\end{abstract}


\section{Introduction}
\IEEEPARstart{M}{odern} power grids span large geographic areas and are becoming more complex due to integration of renewable energy resources and control devices, e.g., inverters~\cite{DOE_report}. 
Solving power flow and optimization problems over many operating scenarios with full network models can become computationally challenging.
 Network reduction techniques are designed to mitigate this complexity by simplifying large networks while preserving essential system characteristics. Unfortunately, key operational characteristics, such as line flows and nodal voltages, are prone to inaccuracies in reduced networks relative to the original networks.
Consequently, network reduction involves a trade-off between accuracy and the level of reduction.
In this paper, we focus on network reduction for steady-state applications, particularly power flow and optimal power flow (OPF) studies.

 Several transmission network reduction methods preserve selected branch flows. For instance, bus aggregation methods cluster nodes with similar power transfer distribution factors (PTDFs) so that injections at each group have comparable effects on flows~\cite{PTDF_2012, PTDF_2015}.
 Similarly,~\cite{benjamin} group neighboring nodes by their contributions to congested lines in a DC OPF problem~\cite{benjamin}.
Optimization-based and data-driven methods tune susceptances or nodal injections so that the reduced model reproduces the DC flow response of the original system~\cite{Taheri_parameter_learning_2024, Zhu_Optimization_DC_NR_2018, Ribeiro_DD_TNR_2024}.
These methods are practical when preserving flows on a set of branches is the primary objective. However, methods based on PTDFs or DC power flow generally neglect voltage deviations and have limited applicability in distribution feeders~\cite{Bolognani_2016}.

Another common reduction approach keeps a subset of nodes and replaces the remainder of the network with an equivalent representation. Classical examples include Ward~\cite{ward_1949} and Kron reduction~\cite{kron_2013}. For transmission networks, this idea is used to keep boundary buses and reduce internal buses within zones to accelerate planning studies~\cite{TNR_multicut_kron_2018,jiang_zone_2023}.
For distribution feeders, inversion reduction preserves voltages at predefined kept nodes by computing equivalent impedances and distributed injections~\cite{Zachary_inversion_reduction_2018,Zachary_inversion_reduction_2019,Zachary_inversion_reduction_2021}. However, these distributed injections obscure the mapping between reduced nodes and physical loads or controllable devices.
Another method reduces distribution feeders into two-node equivalents to preserve voltage drop and total losses~\cite{todorovski_2024_dnr}, but does not preserve detailed nodal voltage profiles.
While these approaches are effective when representative nodes are specified \emph{a priori}, relying on a predetermined selection inherently restricts reduction quality. An optimally chosen set of kept nodes could achieve a lower approximation error for a given network size, or maintain the same error using a smaller subset of nodes.

To address the sub-optimality of fixed node selection, several recent works employ heuristic rules to group or select nodes based on network structure, electrical distance, or operational characteristics.
For transmission network reduction,~\cite{meyn_2025_ccsc} clusters nodes so that coherent generators remain in the same groups.
The amount of power passing through each node has also been used to prioritize representative nodes~\cite{huang_2024_pesgm}.
For distribution networks,~\cite{Mads_energies} clusters nodes using electrical distance, selects representative nodes, and then applies Kron reduction to form the reduced network.~\cite{huang_2024_drl} iteratively removes leaf nodes and updates reduced network parameters using reinforcement learning. Similarly,~\cite{graine_2023_dnr} reduces leaf nodes to simplify the distribution network reconfiguration problem. In these methods, node or cluster selection is not tuned against an explicit error metric, and the resulting reduction error is typically assessed a posteriori. Thus, they may fail to meet the accuracy requirements of the target application.

Unlike previous heuristic approaches, \cite{optiKRON2022} introduced Opti-KRON, a mixed-integer linear program (MILP) that selects kept nodes to optimally balance reduction levels with voltage accuracy.
However, the original framework in~\cite{optiKRON2022} lacked the scalability required for practical distribution grids and failed to preserve radial topologies.

We extend~\cite{optiKRON2022} to establish a scalable, optimization-based framework for balanced distribution feeders represented by single-phase equivalents. Our framework accommodates multiple voltage levels (e.g., medium and low voltage) and applies to both passive and active radial distribution networks. The main contributions of this work are:

\begin{itemize}
\item \textbf{Scalability:} We overcome the previous limitations to handle realistic feeders spanning thousands of nodes. This is achieved via a two-stage strategy: first, eligible zero-injection nodes are reduced, followed by an MILP optimization utilizing a cutting-plane restriction and tightened Big~M constraints.
\item \textbf{Explicit Error Bound:} We develop linearized voltage magnitude error constraints to explicitly bound and enforce a user-defined error threshold during the optimization process.
\item \textbf{Radiality Preservation:} Because Kron reduction typically transforms radial networks into dense, meshed equivalents~\cite{kron_2013}, we introduce an exact radialization procedure. By identifying and reinserting a minimal set of previously reduced nodes, we reconstruct an equivalent radial feeder. Preserving this topology is critical, as it enables algorithms developed for radial systems~\cite{dist_flow,BFS}, supports convex conic power flow formulations~\cite{Jabr_radial}, and maintains computational sparsity.
\end{itemize}

\begin{table*}[t]
\caption{Comparison of representative network reduction methods and the proposed framework.}
\label{tab:lit_compare}
\centering
\setlength{\tabcolsep}{4.2pt}
\renewcommand{\arraystretch}{1.18}
\begin{tabular}{C{2.35cm} C{1.45cm} C{1.55cm} C{2.6cm} C{1.15cm} C{1.65cm} C{1.65cm}}
\hline
\textbf{Reference} &
\textbf{Network} &
\begin{tabular}[c]{@{}c@{}}\textbf{Node Selection}\\\textbf{approach}\end{tabular} &
\begin{tabular}[c]{@{}c@{}}\textbf{Node Selection}\\\textbf{basis}\end{tabular} &
\textbf{Scalability} &
\begin{tabular}[c]{@{}c@{}}\textbf{Radiality}\\\textbf{preserved}\end{tabular} &
\begin{tabular}[c]{@{}c@{}}\textbf{Explicit}\\\textbf{error bound}\end{tabular} \\
\hline

\cite{PTDF_2012,PTDF_2015}
& Transmission
& Heuristic
& PTDF similarity
& Yes
& N/A
& No \\

\cite{benjamin}
& Transmission
& Heuristic
& Congestion
& Yes
& N/A
& No \\

\cite{Taheri_parameter_learning_2024,Zhu_Optimization_DC_NR_2018,Ribeiro_DD_TNR_2024,TNR_multicut_kron_2018,jiang_zone_2023}
& Transmission
& Pre-set
& Zones
& Yes
& N/A
& No \\

\cite{meyn_2025_ccsc,huang_2024_pesgm}
& Transmission
& Heuristic
& Network structure
& Yes
& N/A
& No \\

\cite{Zachary_inversion_reduction_2018,Zachary_inversion_reduction_2019,Zachary_inversion_reduction_2021,todorovski_2024_dnr}
& Distribution
& Pre-set
& Nodes
& Yes
& Yes
& No \\

\cite{Mads_energies}
& Distribution
& Heuristic
& Electrical distance
& Yes
& No
& No \\

\cite{huang_2024_drl,graine_2023_dnr}
& Distribution
& Heuristic
& Network structure
& Yes
& Yes
& No \\

\cite{optiKRON2022}
& General
& Optimization
& Voltage 
& No
& No
& No \\

\cite{optikron3}
& Distribution
& Heuristic
& Voltage 
& Yes
& Yes
& Yes \\

\textbf{This paper}
& \textbf{Distribution}
& \textbf{Optimization}
& \textbf{Voltage}
& \textbf{Yes}
& \textbf{Yes}
& \textbf{Yes} \\
\hline
\end{tabular}
\end{table*}

Table~\ref{tab:lit_compare} compares the proposed framework with representative prior reduction methods. Finally, we explicitly delineate this work from our recent extension to unbalanced three-phase feeders in~\cite{optikron3}. While the present paper develops a rigorous, exact MILP-based optimization framework for single-phase equivalent models,~\cite{optikron3} replaces the exact MILP step with a heuristic exhaustive search procedure to manage the increased computational complexity inherent to multiphase systems.
An overview of this manuscript's optimal network reduction framework, including radialization, is illustrated in Fig.~\ref{fig: overview}.
\begin{figure*}[t]
    \centering
    \includegraphics[width=1\linewidth]{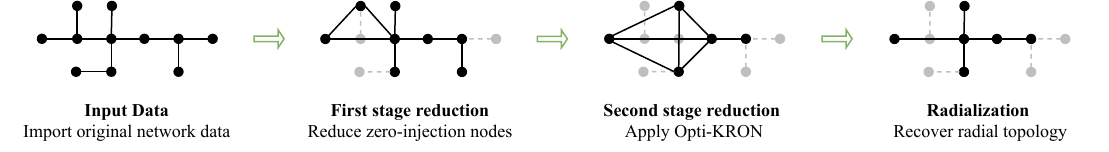}
    \caption{Flow of the proposed network reduction framework. The first stage reduces the network by removing nodes with no active elements (e.g., load or generation). The second stage iteratively applies the proposed optimization problem, i.e., Opti-KRON, to the initial reduced network until further reduction would violate voltage deviation limits. Finally, if the reduced network is meshed, radiality is restored by reintroducing a minimal set of previously reduced nodes. Radialization may also be applied after the first stage reduction.}
    \label{fig: overview}
\end{figure*}

The remainder of this paper is organized as follows. In Section~\ref{sec:model}, we summarize the network model and Kron reduction. Section~\ref{sec: formulation} formulates a mixed-integer nonlinear program (MINLP) for optimal Kron-based network reduction. Reformulation to MILP and the first stage reduction are presented in Section~\ref{sec: scalability}. We introduce the concept of radialization in Section~\ref{sec: radialization}. Experimental results are presented in Section~\ref{sec: results}. Finally, the paper concludes in Section~\ref{sec: end} with a summary and a brief discussion on future directions and applications.

\section{Network model and Kron Reduction} \label{sec:model}

\subsection{Network Representation}
Consider a power system network as a graph $\mathcal{G} = (\mathcal{V}, \mathcal{E})$, where $\mathcal{V}$ and $\mathcal{E}$ represent the set of $n$ electrical buses and $m$ branches, respectively.
The network's structure is defined by the adjacency matrix ${\Lambda} \in {\{0,1\}}^{n \times n}$, where $\Lambda_{ij}=1$ $(\Lambda_{ij}=0)$ indicates the presence (absence) of an edge between nodes~$i$ and~$j$. The nodal admittance matrix, $Y \in \mathbb{C}^{n \times n}$, captures the electrical characteristics of the network. Let $V, I \in \mathbb{C}^{n}$ denote complex nodal voltage and current injection vectors, respectively, which are related via Kirchhoff's Current Law as
\begin{align} \label{eq: KCL}
     I = Y V.
\end{align}

\subsection{Kron Reduction of Electrical Networks} \label{kron-reduction}
The Kron reduction of a graph results in a smaller graph by eliminating a subset of nodes with no power injection, while preserving the equivalent impedances among the remaining nodes. 
Consider a network with an admittance matrix $Y$, where we divide the nodes into a set of nodes to be kept, denoted as $k$, and a set of nodes to be removed, denoted as $r$. If the current injection at the nodes in set $r$ is zero (i.e., $I_r = 0$), we can partition \eqref{eq: KCL} as:
\begin{align} \label{eq: Y-partitioned}
            \left[\begin{array}{c}
            {I}_{k}\\
            \hline  0
            \end{array}\right] & =\left[\begin{array}{c|c}
            Y_{kk} & Y_{kr}\\
            \hline Y_{rk} & Y_{rr}
            \end{array}\right]\left[\begin{array}{c}
            {V}_{k}\\
            \hline {V}_{r}
            \end{array}\right].
\end{align}
The Kron reduction of $Y$ is the Schur complement of $Y_{rr}$~\cite{kron_2013}. Thus, the Kron-reduced admittance matrix, denoted as $Y_{\rm Kron}$, is given by
\begin{align} \label{eq: Y_kron}
     Y_{\rm Kron} &= Y_{kk} - Y_{kr} {Y_{rr}}^{-1} Y_{rk},
\end{align}
such that
\begin{align} \label{eq: Ik=Yk Vk}
     I_k = Y_{\rm Kron} V_k.
\end{align}

\subsection{Kron-based Network Reduction}

While traditional Kron reductions require reducing over a set of zero-injection nodes, the proposed Kron-based approximation allows any set of nodes to be Kron-reduced.
The framework selects an optimal subset of representative nodes, referred to as \textit{super-nodes}, and assigns each reduced node to one of them.
Nodes assigned to the same super-node form one cluster. Within each cluster, we aggregate the current injections of reduced nodes at the associated super-node. The voltage of each reduced node is then approximated by the voltage of its super-node in the reduced network.

To form clusters, we define a binary assignment matrix $A \in \{0,1\}^{n \times n}$, where $A_{i,j}=1$ if node $j$ is represented by super-node $i$, and $A_{i,j}=0$ otherwise. Specifically, node $i$ is a super-node if  $A_{i,i}=1$ and reduced otherwise. For any super-node $i$, the corresponding cluster is defined as $\mathbf{A}_i := \{j \in \mathcal{V} \mid A_{i,j}=1\}$. Each cluster is a connected subgraph: every reduced node in the cluster remains connected to its super-node through nodes belonging to that same cluster.
Fig.~\ref{fig: aggregation} provides an illustrative example of how super-nodes are formed during the network reduction process.
\begin{figure} 
    \centering
    \includegraphics[width=0.95\linewidth]{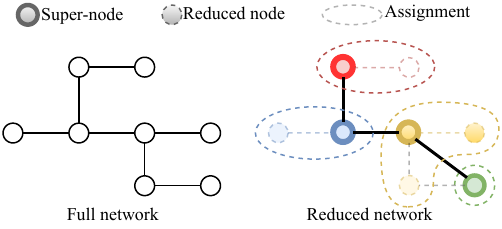}
    \caption{Opti-KRON procedure; an 8-bus radial network has been reduced to a 4-bus network. The injections of reduced nodes are assigned to their super-nodes in each cluster. Reduced network topology and the connections among super-nodes can be found through equation~\eqref{eq: Y_kron}.}
    \label{fig: aggregation}
\end{figure}

Since current can only be assigned to one super-node, $\sum_i A_{i,j}=1$ is always enforced. Furthermore, $A_{i,j} \leq A_{i,i} $ ensures that injections cannot be assigned to a reduced node. Now, we can calculate the aggregated current vector as
\begin{align} \label{eq: Kron_current}
 I_{\rm Kron} = A I,
\end{align}
where $I_{\rm Kron} \in \mathbb{C}^n$ maps the current injections corresponding to reduced nodes to zero. 
In this study, we assume loads and generations have the same current injections in reduced and full networks. Thus, the nodal voltages of the Kron-based reduced network can be obtained from
\begin{align} \label{eq: power-flow-basic}
 Y V_{\rm agg} = A I,
\end{align}
where $V_{\rm agg} \in \mathbb{C}^{n}$ is the voltage of the full network under the aggregated current injections. Since voltages of kept nodes in the full and reduced networks are equal,
without using~\eqref{eq: Y_kron} and~\eqref{eq: Ik=Yk Vk}, we can still compute the super-nodes' voltages resulting from a Kron-based reduction, i.e.,~\eqref{eq: power-flow-basic}.
Additionally, we represent the voltages at Kron-reduced nodes by those at their respective (kept) super-nodes. This assignment is straightforward using matrix $A$ as 
\begin{align} \label{eq: Kron_voltage}
V_{\rm Kron} = A^{\rm T} V_{\rm agg},
\end{align}
where $V_{\rm Kron} \in C^n$ defines nodal voltages of a Kron-based reduced network. The goal of this manuscript is, therefore, to select a nodal assignment matrix $A$ that maximizes the number of reduced \textit{children-nodes}, balanced against how accurately the selected super-nodes represent their respective children-nodes.

\section{Optimal Kron-based Reduction of Networks} \label{sec: formulation}
 
In this section, we formulate an optimization problem that determines the optimal assignment matrix $A$, which controls the trade-off between accuracy and reduction by specifying the cluster configurations and the number of super-nodes. Therefore, we need to define an objective function that balances reduction and voltage accuracy.

To quantify accuracy, we define an error metric that measures the deviation between a node’s voltage in the full network and the voltage of its super-node in the reduced network.
Additionally, to consider different operating conditions we calculate the error with respect to $n_L$ different loading scenarios $l \in \mathcal{L}$. Data for these scenarios comes from representative AC power flow solutions on the full network. We denote the error matrix associated with loading scenario $l$ as
\begin{align} \label{eq: error}
 {E_{l}} = \overbrace{A {\rm diag}\{\hat{V}_l\}}^{\text{Original voltages}} - \overbrace{{\rm diag}\{V_l \} A}^{\text{ Assigned voltages}} \in \mathbb{C}^{n\times n},  
\end{align}
where $\hat{V}_l$ denotes the vector of voltage data in the full network satisfying $Y \hat{V}_l = \hat{I}_l$. Voltage vector $V$ after nodal aggregation is defined by $Y V_l = A \hat{I}_l$ \footnote{For notational simplicity, we hereafter use $V$ in place of $V_{\rm agg}$.}.
Based on \eqref{eq: error}, if $A_{i,j} = 1$, then $E_{l,i,j} =\hat{V}_{l,j} - V_{l,i} $, while $A_{i,j} = 0$ returns $E_{l,i,j} = 0$. To evaluate the worst-case deviation within each cluster, we define the Maximum Intra-Cluster Error (MICE) as
\begin{align} \label{eq: MICE}
 {\rm MICE}_{l,i} = \left\Vert {\vec e_i^{\,T}} E_{l} \right\Vert_{\infty},
\end{align}
where $\vec{e}_i$ represents the $i^{\rm th}$ standard basis vector with all elements zero except for the $i^{th}$ element.
Accordingly, the objective function that optimally balances the trade-off between reduction levels and voltage deviation is then given by:
\begin{align} \label{eq: obj}
\mathbf{O} =\overbrace{ \sum_{l \in \mathcal{L}} {\sum_{i \in \mathcal{V}} { 
 \left\Vert {\vec e_i^{\,T}} E_{l} \right\Vert_{\infty}} } }^{\rm Accuracy} - \alpha  \overbrace{ \sum_{i}^n(1-A_{i,i})}^{\text{ Reduction lev.}},
\end{align}
with $\alpha$ denoting a weighting factor that determines the network reduction priority.  
The term $\sum_{i}^n(1-A_{i,i})$ counts the number of eliminated nodes, as $A_{i,i}=0$ indicates a reduced node. Given this objective function, we can express the optimization problem as follows:
	\begin{subequations} \label{raw model}
        \begin{align}
    \min_{A, V, E} \quad &{\sum_{l \in \mathcal{L}} {\sum_{i \in \mathcal{V}} { \left\Vert {\vec e_i^{\,T}} E_{l} \right\Vert_{\infty} }}  - \alpha \sum_{i \in \mathcal{V}} (1- A_{i,i}) }  \label{eq: obj2}\\
    {\rm s.t.}\;\; \quad & {E_{l}} = A {\rm diag}\{\hat{V}_l\} - {\rm diag}\{V_l \} A \quad \forall l \in \mathcal{L} \label{eq: error-opt}\\
    & \mathbf{Y} V_l = A \hat{I_l} \quad \forall l \in  \mathcal{L} \label{eq: powerflow}\\
    & - \bar{E} \leq |A^{\rm T}V_l|-|\hat{V}_l| \leq \bar{E} \quad \forall l \in \mathcal{L} \label{eq: v_abs}\\   
    & A^{\rm T} \mathbf{1} = \mathbf{1} \label{eq: assign_limit1}\\
    &A_{i,j} \leq A_{i,i} \quad \forall i,j \in \mathcal{V} \label{eq: assign_limit2}\\
    &A_{i,j} \leq \Lambda_{i,j} \quad \forall i,j \in \mathcal{V}, i\neq j\label{eq: assign_adj}\\
    &A_{i,j} \in \{0,1\} \quad \forall i,j \in \mathcal{V}. \label{eq :binary}
	\end{align}
 \end{subequations}
In the proposed model, we can calculate super-node voltages for each loading scenario in~\eqref{eq: powerflow}. Equation~\eqref{eq: v_abs} bounds the voltage magnitude deviation to lie within $\Bar{E}$.
Here, $\bar{E}$ is the user-defined error limit for voltage magnitude.
Each node can only be assigned to one cluster, forced by~\eqref{eq: assign_limit1}.
According to~\eqref{eq: assign_limit2}, reduced nodes can be assigned only to super-nodes.
To avoid cases where the optimizer aggregates nodes that are far from each other, we use the adjacency matrix to restrict nodal aggregations to neighboring nodes with~\eqref{eq: assign_adj}. The optimization converges when further reduction violates~\eqref{eq: v_abs}. Therefore, we should set $\alpha$ to a positive value large enough to incentivize further reduction, but not excessively large that the reduction term dominates the objective.
However,~\eqref{raw model} represents a non-convex MINLP, and solving it for realistic networks presents computational challenges. Thus, in the following section, we enhance the computational scalability of the network reduction problem.
 
\section{Improving Scalability} \label{sec: scalability}
To improve the scalability of~\eqref{raw model}, we introduce a Kron-based reduction algorithm to selectively reduce a subset of the zero-injection nodes before solving the optimal network reduction problem. In the next step, we reformulate~\eqref{raw model} as an MILP. Finally, we add a cutting plane to speed up the MILP.
\subsection{Reduction of Zero-Injection Nodes}
At this step, we reduce a subset of zero-injection nodes, defined as $\mathcal{Z} := \{i \in \mathcal{V} \mid \hat{I}_{l,i} =0, \ \forall l \in \mathcal{L}\}$. Reduction of these nodes, as shown in~\eqref{eq: Y-partitioned}, does not affect the voltages at the other nodes. 
For each zero-injection node satisfying the conditions described below, we aim to assign it to a super-node from the set of non-zero-injection nodes, denoted by $\mathcal{I} := \mathcal{V} \setminus \mathcal{Z}$, and update the assignment matrix $A$ accordingly. These assignments must satisfy the following conditions:
\begin{enumerate}
 \item  The voltage magnitude difference between a zero-injection node and its super-node must not exceed the threshold $\bar{E}$, to satisfy~\eqref{eq: v_abs}.
 \item The path between each reduced node and its super-node must contain only zero-injection nodes from the same cluster and no other super-nodes.
\end{enumerate}
However, not all the nodes in $\mathcal{Z}$ will necessarily satisfy these conditions.
 Algorithm~\ref{alg: zero_injection_reduction} summarizes the procedure to assign eligible zero-injection nodes to super-nodes, where each node $i \in \mathcal{Z}$ is associated with a candidate set $\mathcal{C}_i := \{ j \in \mathcal{I} | \pi_{i,j} \subseteq \mathcal{Z}\}$, with $\pi_{i,j}$ representing the set of nodes along the path from node $i$ to node $j$, excluding node $j$.
 
\begin{algorithm}[H]
    \caption{Assignment of Zero-Injection Nodes}
    \begin{algorithmic}[1]
    
    \STATE \textbf{Input:} $\Lambda$ (adjacency matrix), $\hat{V}$ (voltage data) , $\bar{E}$ (voltage error bound), $\mathcal{C}$ (candidate node set), $\mathcal{Z}$ (zero-injection node set)
    \STATE $A \gets I$ \hfill $\triangleright$ Assign the Identity
    
    \STATE $\mathcal{U} \gets \mathcal{Z}$ \hfill $\triangleright$ Unassigned nodes
    
    \STATE $\Delta \gets 1$ \hfill $\triangleright$ Termination indicator

    \STATE  $d_{i,j} \gets \max_{l \in \mathcal{L}} ||\hat{V}_{l,i}| - |\hat{V}_{l,j}||, \ \forall i \in \mathcal{Z}, \ j \in \mathcal{C}_i$

    \WHILE{ $\Delta = 1$ }
        \STATE $\Delta \gets 0$
        \FOR{each $i \in \mathcal{U}$}
            \FOR{each $j \in \mathcal{V} \setminus \mathcal{U}$ such that $\Lambda_{i,j} = 1$}
                \STATE $k \gets \{ m \in \mathcal{V} \mid A_{m,j} = 1 \}$
                \IF{$d_{i,k} \leq \bar{E}$ and $d_{i,k} \leq d_{i,h} \ \forall h \in \mathcal{C}_i$}
                    \STATE $A_{i,i} \gets 0$, \quad $A_{k,i} \gets 1$
                    \STATE $\mathcal{U} \gets \mathcal{U} \setminus \{ i \}$, \quad $\Delta = 1$
                \ENDIF
            \ENDFOR
        \ENDFOR
    \ENDWHILE
    \end{algorithmic}
    \label{alg: zero_injection_reduction}
\end{algorithm}
After this process, the proposed optimization problem, finalized later in this section as~\eqref{opt_problem_final}, can be applied to the resulting network to yield an optimal reduced network.
\subsection{Problem Reformulation}
 To reformulate the non-linear term ${\rm diag}\{V_l \}A$
 in~\eqref{eq: error-opt}, we employ the Big~M method by defining an auxiliary variable $W_l := {\rm diag}\{V_l \}A$. This allows us to replace~\eqref{eq: error-opt} with its equivalent set of inequality constraints:
\begin{subequations} \label{eq: Big-M}
    \begin{align}
    &W_{l,i,j} \leq (1-A_{i,j})M + V_{l,i} \quad \forall i,j \in \mathcal{V}, \ \forall l \in \mathcal{L}\\
    &W_{l,i,j}  \geq (A_{i,j}-1)M + V_{l,i} \quad \forall i,j \in \mathcal{V}, \ \forall l \in \mathcal{L} \\
    &W_{l,i,j}  \leq MA_{i,j} \quad \forall i,j \in \mathcal{V}, \ \forall l \in \mathcal{L}  \\
    &W_{l,i,j}  \geq-MA_{i,j} \quad \forall i,j \in \mathcal{V}, \ \forall l \in \mathcal{L}.
    \end{align}
\end{subequations}
The parameter $M$ in \eqref{eq: Big-M} represents a sufficiently large constant. However, selecting excessively large values can cause numerical challenges during the optimization process~\cite{bigM}. Hence, it is valuable to determine the smallest possible $M$ that guarantees the feasibility of the constraints without compromising numerical performance. Appendix \ref{appendix: big-m} discusses the selection of $M$ in more detail.
Now, we can update~\eqref{eq: error-opt} to
\begin{align}
    {E_{l}} = A {\rm diag}\{\hat{V}_l\} - W_l       \quad \forall l \in \mathcal{L}. \label{eq: error-opt_lin}
\end{align}

To derive a linear approximation of~\eqref{eq: v_abs}, we first calculate the nodal voltage error as
\begin{align} \label{eq: errordef}
    e_{l} = A^{\rm T} V_l -\hat{V}_l.
\end{align}
We can rewrite~\eqref{eq: errordef}  as
     $A^{\rm T} V_l = \hat{V}_l + e_{l}$, and update~\eqref{eq: v_abs} to
\begin{align}
    &|\hat{V}_l|-\Bar{E} \leq |\hat{V}_l + e_{l}|  \leq |\hat{V}_l| + \Bar{E} \quad \forall l \in \mathcal{L}. \label{eq: v_abs^1}
\end{align}
In practice, the voltage magnitudes $|\hat{V}_l|$ are typically close to 1~p.u. and $0 \leq \bar{E} \ll 1$, which implies that $|\hat{V}_l|-\bar{E} \ge 0$. Therefore, we can take the square of~\eqref{eq: v_abs^1} and simplify it to
{\small
 \begin{subequations} \label{eq: v_abs^2}
    \begin{align}
        {-2|\hat{V}_l|\bar{E} + \bar{E}^2} - |e_{l}|^2 &\le
        2({\Re\{\hat{V}_l\} \Re\{e_{l}\} +  \Im\{\hat{V}_l\} \Im\{e_{l}\}}) \label{eq: v_abs^2-ub}\\
        \ {2|\hat{V}_l|\bar{E} + \bar{E}^2} - |e_{l}|^2 &\ge
        2({\Re\{\hat{V}_l\} \Re\{e_{l}\} +  \Im\{\hat{V}_l\} \Im\{e_{l}\}}), \label{eq: v_abs^2-lb}
    \end{align}
\end{subequations}}
where $\Re\{ \cdot\}$ and $\Im \{\cdot \}$ denote real and imaginary parts of complex values.~\eqref{eq: v_abs^2-ub} represents a non-convex constraint.
Considering that we are solving the problem in a per-unit system, $\bar{E}$ and $e_{l}$ represent voltage deviations normalized to the system's base voltage. 
Accordingly, since $ \bar{E}^2 \ll 1 $ and $ |e_l|^2 \ll \mathbf{1} $, we can approximate~\eqref{eq: v_abs^2} with $ |e_l|^2 - \bar{E}^2\approx \mathbf{0}$: 
 \begin{subequations} \label{eq: v_abs_lin}
    \begin{align}
        -& |\hat{V}_l|\bar{E}  \leq
        {\Re\{\hat{V}_l\} \Re\{e_{l}\} +  \Im\{\hat{V}_l\} \Im\{e_{l}\}}\\
        & {|\hat{V}_l|\bar{E} } \geq
        {\Re\{\hat{V}_l\} \Re\{e_{l}\} +  \Im\{\hat{V}_l\} \Im\{e_{l}\} }.
    \end{align}
\end{subequations} 
We consider $\Re\{e_{l}\}$ and $\Im\{e_{l}\}$ as two distinct variables, as shown in Appendix~\ref{appendix: rec-decomp}. Thus,~\eqref{eq: v_abs_lin} represents a linear approximation of~\eqref{eq: v_abs}. In Section~\ref{sec: results}, we show that the approximation error is negligible for realistic networks.

\subsection{Cutting Plane Restriction}
The following cutting plane reduces the feasible set by limiting the number of simultaneous reductions in the network to no more than $q\ge 1$:
\begin{align}\label{eq: CP init}
\sum_{i \in \mathcal{V}}(1-A_{i,i}) \leq  q.
\end{align}
Adding~\eqref{eq: CP init} to the MILP problem results in a more tractable optimization problem. However, it limits the degree of reduction, which necessitates an iterative approach to reduce networks further.
In an iterative approach, the solution at each iteration builds on the optimal solution from the previous iteration. Thus, we need to integrate previous assignments, as determined from $A$ at the optimal point. First, we need to define a new variable $\Omega \in \{0,1\}^{n \times n}$ that captures only the current iteration’s assignments and adjust \eqref{eq: assign_adj} as
\begin{align}
    &\Omega_{i,j} \leq \Lambda_{i,j} \quad \forall i,j \in \mathcal{V}, i\neq j.\label{eq: assign_adj_new}   
\end{align}
Then, the full assignment matrix can be computed from 
 \begin{align}
    A = \Omega A^{\rm prev}, \label{eq: A_Omega}
\end{align}
where $A^{\rm prev}$ represents the optimal $A$ from the previous iteration. For the initial iteration, $A^{\rm prev}$ is derived from the first stage (Algorithm~\ref{alg: zero_injection_reduction}), or, if stage 1 is skipped, it is set to the identity matrix $I$.
To avoid repeatedly considering reduced nodes, we should also update the objective function to
\begin{align} \label{eq: obj-updated}
\mathbf{O} = \sum_{l \in \mathcal{L}} {\sum_{i \in \mathcal{V}} { 
 \left\Vert {\vec e_i^{\,T}} E_{l} \right\Vert_{\infty}}}  - \alpha  \sum_{i \in \mathcal{V}}(A_{i,i}^{\rm prev}-A_{i,i}),
\end{align}
and~\eqref{eq: CP init} as
\begin{align}\label{eq: CP}
\sum_{i \in \mathcal{V}}(A_{i,i}^{\rm prev}-A_{i,i}) \leq  q.
\end{align}
Note that in the three-phase extension~\cite{optikron3}, the exhaustive-search formulation is restricted to the single-reduction case, i.e., \(q=1\).
In contrast, the MILP framework developed here allows a user-selected reduction limit \(q\), and thus multiple simultaneous reductions per iteration.

Now, the proposed optimization problem, which we refer to as Opti-KRON, can be stated as
\[
\begin{tabular}{@{}p{\linewidth}@{}}
\toprule
\textbf{Opti-KRON}\\
\hline
\end{tabular}
\]
{\small
\begin{equation} \label{opt_problem_final}
\hspace*{-\leftmargini}%
\begin{alignedat}{2}
\min_{\substack{A,V,E \\ \Omega \in \{0,1\}^{n \times n}}}& \
 \sum_{l\in\mathcal L}\sum_{i\in\mathcal V}\!\left\Vert \vec e_i^{\,T} E_l \right\Vert_\infty
  - \alpha \sum_{i=1}^n\!\big(A^{\rm prev}_{i,i}-A_{i,i}\big) \ & & \\
\text{s.t.} \ 
& \eqref{eq: powerflow} && \mathllap{\text{\small (Kirchhoff's Current Law)}} \\
& \eqref{eq: assign_limit1}, \eqref{eq: assign_limit2}, \eqref{eq: assign_adj_new}, \eqref{eq: A_Omega}
&\quad & \mathllap{\text{ (Assignment constraints)}} \\
& \eqref{eq: Big-M}, \eqref{eq: error-opt_lin}
&& \mathllap{\text{(Error matrix ($E_l$) calculation)}} \\
& \eqref{eq: errordef}, \eqref{eq: v_abs_lin}
&& \mathllap{\text{(Error bounds)}} \\
& \eqref{eq: CP}
&& \mathllap{\text{(Cutting plane)}}
\end{alignedat}
\end{equation}
}
\noindent\rule{\dimexpr\linewidth-2.3pt}{0.3pt}
The formulation in~\eqref{opt_problem_final} is presented in complex variables. To reformulate using real values only, we apply the rectangular coordinate transformation. Details of the reformulation are presented in Appendix~\ref{appendix: rec-decomp}. Note that~\eqref{opt_problem_final} represents an MILP with $\mathrm{tr}(A^{\rm prev}) + \mathbf{1}^{\rm T} \Lambda \mathbf{1}$ integer decision variables, where at most $q$ node reductions are allowed by~\eqref{eq: CP}. Algorithm~\ref{alg: NR} describes an iterative approach to implement~\eqref{opt_problem_final}.
\begin{algorithm}[H]
	\caption{Network Reduction Algorithm}
    \begin{algorithmic}[1]
	\STATE \textbf{Input:} $Y$, $\hat{V}$, $\hat{I}$, $\Lambda$, $A^{\rm prev}$, $\alpha$, $q$, $\bar{E}$
			
            \STATE  $\Delta \gets \infty$  \hfill $\triangleright$ Initialize termination indicator
            
             \WHILE{$\Delta > 0$ }
                    \STATE $A^* \gets $ Solve optimization problem~\eqref{opt_problem_final} 
                    \STATE $\Delta \gets  \sum_i (A^{\text{prev}}_{i,i} - A_{i,i}^*)$ 
                        \STATE $\Lambda \gets (A^* \Lambda {A^*}^{\rm T}) \circ (\mathbf{1} \mathbf{1}^{\rm T} - I)$
                        \STATE  $A^{\text{prev}} \gets A^*$
                \ENDWHILE
    \end{algorithmic}
            \label{alg: NR}
\end{algorithm}

At each iteration, the adjacency matrix $\Lambda$ is updated to connect each super-node to all neighbors of itself and its assigned child-nodes through $ (A^* \Lambda {A^*}^{\rm T}) \circ (\mathbf{1} \mathbf{1}^{\rm T} - I)$, where $\circ$ denotes element-wise multiplication.  

Algorithm~\ref{alg: NR} terminates after finitely many iterations; at each non-terminating iteration at least one super-node is reduced ($\Delta>0$). Therefore, the number of super-nodes strictly decreases with each non-terminating iteration. Since the total node count is finite and strictly bounded below, the algorithm is guaranteed to terminate. Moreover, when $\alpha$ is chosen large enough to prioritize any feasible additional reduction, termination with $\Delta=0$ occurs only when no further node reduction is feasible under the constraints in~\eqref{opt_problem_final}. In practice, this stopping condition is mainly controlled by the error threshold~$\bar{E}$, since additional reductions become infeasible once they would violate the prescribed error bound.

The assignment matrix of the last iteration reveals the super-nodes, reduced nodes, and cluster configurations of the optimal Kron-based reduced network. Crucially, this explicit mapping ensures that physical devices, such as loads and distributed resources, remain strictly associated with known locations in the reduced model, preserving interpretability for downstream optimization tasks. However, the Kron-based formulation will produce a reduced, but meshed, representation of the original network. To recover a radial distribution network, we propose a novel \textit{radialization} procedure next.

\section{Radialization} \label{sec: radialization}
In this section, we develop a new \textit{radialization} methodology for recovering equivalent radial networks from Kron-reduced (meshed) networks. Radialization effectively excises all connected, meshed subgraphs of super-nodes of the reduced network and replaces them by electrically equivalent radial subgraphs. This process is defined next and depends on the following definition and lemma.
\begin{definition}[Clique]
A subset of nodes $C \subseteq \mathcal{V} $ in a graph $G = (\mathcal{V}, \mathcal{E}) $ is called a \emph{clique} if every pair of distinct nodes in $C $ is connected by an edge, i.e., $(i, j) \in \mathcal{E} $ for all $i, j \in C $, $i \neq j $.
\end{definition}

\begin{lemma} \label{lem: cliques}
Let $G = (\mathcal{V}, \mathcal{E}) $ be a graph representing an electrical network with complex admittance matrix $Y^{\mathrm{full}} \in \mathbb{C}^{n \times n} $. Suppose node $r \in \mathcal{V} $ is Kron-reduced. Then:
\begin{enumerate}
    \item All nodes that were adjacent to $r $ in the original network become directly connected to each other in the reduced network and form a clique.
    \item The connections among nodes that were not adjacent to $r $ remain unchanged.
\end{enumerate}
\end{lemma}
\begin{proof}
Let $\mathcal{K} \subset \mathcal{V}$ be the set of nodes we \emph{keep} in the reduced network. We further partition $\mathcal{K}$ into two disjoint subsets:
$\mathcal{K}_{\text{a}} = \left\{ i \in \mathcal{K} \mid i \text{ is adjacent to } r \right\}$ and
$\mathcal{K}_{\text{f}} = \left\{ i \in \mathcal{K} \mid i \text{ is not adjacent to } r \right\}$, where ``f" signifies ``far". Without loss of generality, we can partition the admittance matrix as
\begin{align}
  Y^{\text{full}}
  \;=\;
  \begin{bmatrix}
    Y_{\mathcal{K}_{\text{f}},\mathcal{K}_{\text{f}}} 
      & Y_{\mathcal{K}_{\text{f}},\mathcal{K}_{\text{a}}} 
      & \mathbf{0} \\
    Y_{\mathcal{K}_{\text{a}},\mathcal{K}_{\text{f}}}  
      & Y_{\mathcal{K}_{\text{a}},\mathcal{K}_{\text{a}}}
      & Y_{\mathcal{K}_{\text{a}},r} \\
    \mathbf{0} 
      & Y_{r,\mathcal{K}_{\text{a}}} 
      & Y_{r r}
  \end{bmatrix},
\end{align}
where the zero blocks reflect that no node in $\mathcal{K}_{\text{f}} $ is adjacent to $r$. Here, $Y_{rr} $ is a scalar, and $Y_{\mathcal{K}_{\text{a}},r} $ and $Y_{r,\mathcal{K}_{\text{a}}} $ are column and row vectors with no zero entries. Now, considering~\eqref{eq: Y_kron}, the Kron reduction of $Y^{\rm full}$ with respect to $r$ can be stated as
\begin{align}
  Y^{\text{Kron}} 
  &=
  \begin{bmatrix}
    Y_{\mathcal{K}_{\text{f}},\mathcal{K}_{\text{f}}} 
      & Y_{\mathcal{K}_{\text{f}},\mathcal{K}_{\text{a}}}\\
    Y_{\mathcal{K}_{\text{a}},\mathcal{K}_{\text{f}}}  
      & Y_{\mathcal{K}_{\text{a}},\mathcal{K}_{\text{a}}}
  \end{bmatrix}
  -
  \frac{
  \begin{bmatrix}
    \mathbf{0}\\
    Y_{\mathcal{K}_{\text{a}},r}
  \end{bmatrix}
  \begin{bmatrix}
    \mathbf{0} & Y_{r,\mathcal{K}_{\text{a}}}
  \end{bmatrix}}{Y_{r r}}.
\end{align}
This simplifies to
\begin{align} \label{eq: Kron_adj_far}
  Y^{\text{Kron}} 
  =
  \begin{bmatrix}
    Y_{\mathcal{K}_{\text{f}},\mathcal{K}_{\text{f}}}
      & Y_{\mathcal{K}_{\text{f}},\mathcal{K}_{\text{a}}}\\
    Y_{\mathcal{K}_{\text{a}},\mathcal{K}_{\text{f}}} 
      & Y_{\mathcal{K}_{\text{a}},\mathcal{K}_{\text{a}}} 
      - \frac{Y_{\mathcal{K}_{\text{a}},r} \ Y_{r,\mathcal{K}_{\text{a}}}}{Y_{r r}}
  \end{bmatrix}.
\end{align}

Note that only the lower-right block of~\eqref{eq: Kron_adj_far}, which represents the connection among adjacent nodes, changes. Thus, Kron reduction does not affect connectivity between nodes which are non-adjacent to the reduced node. Additionally, the term $ Y_{\mathcal{K}_{\text{a}},r} Y_{r,\mathcal{K}_{\text{a}}}$ is the outer product of two vectors and is, therefore, a fully dense $\text{degree}(r) \times \text{degree}(r)$ matrix. This implies that any two nodes previously adjacent to node $r$ become directly connected in the reduced network. Thus, the Kron-reduction results in a \textit{clique} of degree~$r$ among all adjacent nodes of a reduced node, regardless of the initial graph's topology (e.g., meshed or radial). 
\end{proof}

Lemma~\ref{lem: cliques} exploits the fact that
$Y_{\mathcal{K}_{\text{f}},\mathcal{K}_{\text{f}}}$ is unaffected by the Kron reduction, as first observed in \cite{Kekatos_graph_id}, and it addresses Kron reduction of a single node. However, since simultaneous Kron reduction of $m$ nodes is equivalent to $m$ sequential single-node reductions~\cite{kron_2013}, the lemma generalizes to reductions of multiple nodes.
To further elucidate the structure of Kron-reduced radial networks, we recall the definition of maximal cliques:
\begin{definition}[Maximal Clique]
A clique is \emph{maximal} if it is not included in any larger clique.
\end{definition}
When the original network is radial, the cliques introduced by Kron reduction are both maximal and edge-disjoint (i.e., they do not share edges)~\cite{inverse_pf}. Maximal cliques with two nodes form simple edges, while those with three or more introduce cycles (i.e., local meshing). Consequently, a radial network may become a dense network after Kron reduction, e.g., see Fig.~\ref{fig: theorem}(b). 
This motivates the question: \textit{Can we reverse the Kron-based transformation and recover a radial representation?} To address this, we replace each maximal clique of three or more nodes by a radial sub-network. Because the maximal cliques are edge-disjoint, these replacements can be performed independently. Replacing all such maximal cliques yields a fully \textit{radialized} reduced network.

Reference~\cite{inverse_pf} leverages the characteristics of the admittance matrix of each maximal clique as a sub-network to recover a less reduced network, which also exhibits radial topology. This approach involves introducing additional hidden nodes to each maximal clique and determining the radial connections in the reduced network. However, the objective in~\cite{inverse_pf} is to reconstruct the full network from a Kron-reduced network in the absence of the full network’s admittance matrix.
In contrast, this paper focuses on identifying the \textit{critical} reduced nodes that must be reinserted to the set of super-nodes to recover radiality. Specifically, we demonstrate how to ensure that each maximal clique with three or more nodes can be transformed into a radial subgraph by reintroducing the minimal subset of previously reduced nodes.

To accomplish this, let $\mathcal{G} := (\mathcal{V}, \mathcal{E})$ again be the full graph representing a radial electrical network, and let $\mathcal{G_R} := (\mathcal{V_R}, \mathcal{E_R})$ denote its Kron-reduced graph. Consider one of the maximal cliques in $\mathcal{G_R}$, and let $\mathcal{V}_{\rm clique} \subseteq \mathcal{V_R} \subset \mathcal{V}$ be the set of its nodes with $|\mathcal{V}_{\rm clique}| \ge 3$. 
We define $\mathcal{T} = (\mathcal{V_T}, \mathcal{E_T})$ as the sub-tree of $\mathcal{G} $ that spans all nodes in $\mathcal{V}_{\rm clique}$, which includes additional nodes from $ \mathcal{V} $ to maintain connectivity, such that $\mathcal{V}_{\rm clique} \subset \mathcal{V_T} $. Since $\mathcal{G}$ is radial, $\mathcal{T}$ is unique~\cite{diestel2018graph}.
\begin{figure}[!t]
    \centering
    \includegraphics[width=1\linewidth]{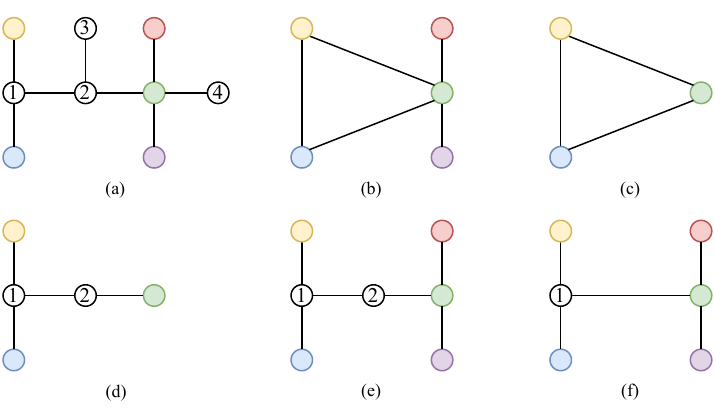}
    \caption{From a radial network toward a reduced radial network. (a) represents a full network $\mathcal{G}$. (b) shows the reduced network $\mathcal{G_R}$. A maximal clique is demonstrated in (c). The spanning sub-tree of the clique $\mathcal{T} $ is represented in (d). (e) shows the augmented Kron-reduced graph $\widetilde{\mathcal{G}}_\mathcal{R}$. (f) demonstrates the reduced radialized graph. Although node 2 has a degree of 3 in $\mathcal{G}$, its degree is 2 in $\mathcal{T} $, and thus it is not a critical node for radiality.} 
    \label{fig: theorem}
\end{figure}
\begin{theorem}\label{thm: radialization}
The nodes in $\mathcal{T}$ with degree $\ge 3$ are the minimal subset of previously reduced nodes that must be kept to ensure radial connectivity among the nodes in $\mathcal{V}_{\rm clique}$.
\end{theorem}
\begin{proof}
To prove Theorem~\ref{thm: radialization}, we first show that not reducing the nodes in $\mathcal{T}$ ensures nodes in $\mathcal{V}_{\rm clique}$ become radially connected. We then show that only nodes of degree $\ge 3$ in $\mathcal{T}$ are critical to preserve the radial structure. 

Consider the Kron-reduced graph $\widetilde{\mathcal{G}}_\mathcal{R} = (\widetilde{\mathcal{V}}_\mathcal{R},\widetilde{\mathcal{E}}_\mathcal{R})$ by replacing the clique on \(\mathcal V_{\rm clique}\) in \(\mathcal G_R\) with the tree \(\mathcal T\). 
Since $\mathcal{T}$ is connected and radial, and none of the nodes in $\mathcal{V_T}$ were reduced, the connections among nodes in $\mathcal{T}$ remain unchanged in $\widetilde{\mathcal{G}}_\mathcal{R}$. This follows from Lemma~\ref{lem: cliques}, which shows that Kron reduction only modifies the connectivity among nodes that are simultaneously adjacent to a reduced node.  
Thus, in $\widetilde{\mathcal{G}}_\mathcal{R}$, the paths between nodes in $\mathcal{V}_{\rm clique}$ lie entirely within $\mathcal{T}$ and remain radial.  
Moreover, by Lemma~\ref{lem: cliques}, any vertex in $\mathcal{V_T} \setminus \mathcal{V}_{\rm clique}$ with degree at most 2 can still be Kron-reduced from $\widetilde{\mathcal{G}}_\mathcal{R}$ while preserving radiality.  
Consequently, only the vertices of degree $\ge 3$ in $\mathcal{T}$ are necessary to preserve radiality among the nodes in the original clique.
\end{proof}

We can use Theorem~\ref{thm: radialization} to find a reduced, yet radial, network from a meshed Kron-reduced network by identifying nodes that are critical for radialization of each maximal clique. For example, consider the reduced network in Fig.~\ref{fig: theorem} (c) as a maximal clique. Its spanning sub-tree in the full network is represented in Fig.~\ref{fig: theorem} (d), where node 1 is the only node with degree $\geq 3$ within the sub-tree. Therefore, Fig.~\ref{fig: theorem} (f) is the radialized network of (d).
In this process, the nodal injections of critical nodes are not reassigned. Thus, super-nodes in the radialized networks have the same voltage profile as in the pre-radialized networks and are equivalent. Also, radialization does not affect the errors of Kron-based network reduction.

Fig.~\ref{fig: overview} summarizes the proposed reduction workflow.
In the first stage, Algorithm~\ref{alg: zero_injection_reduction} eliminates eligible zero-injection nodes. 
The resulting network is then further reduced through the iterative Opti-KRON procedure in Algorithm~\ref{alg: NR}.
Finally, the radialization step identifies maximal cliques created by Kron reduction and reinserts the critical reduced nodes characterized in Theorem~\ref{thm: radialization} to recover a radial reduced feeder.

\section{Numerical Case Studies} \label{sec: results}
    This section presents the simulation results for the proposed model applied to the 533-bus distribution test system~\cite{9124806} and a real single-phase distribution feeder with 3499 nodes, i.e. South Alburgh test case. The optimization problems in this study were implemented in Julia with JuMP~\cite{JuMP} and solved using Gurobi Optimizer 12.0, with the optimality gap set to 0.1\%. To evaluate reduction quality, we compare voltage magnitudes resulting from power flow on the reduced networks with those from the full networks. We also compare our method with two approaches that have been used in the power systems literature on network reductions.

\subsection{533-Bus Distribution Feeder}
This test system, available in MATPOWER, is based on a real 12~kV distribution feeder in southern Sweden and is represented as a single-phase equivalent of a balanced three-phase system. The network includes 262 nodes with distributed generation, with installed capacities up to 2 MW. We consider two representative operating scenarios corresponding to the maximum and minimum net load hours of 2022.~Fig. \ref{fig: v-full} represents the voltage profile for these scenarios.
\begin{figure}
    \centering
    \includegraphics[width=1\linewidth]{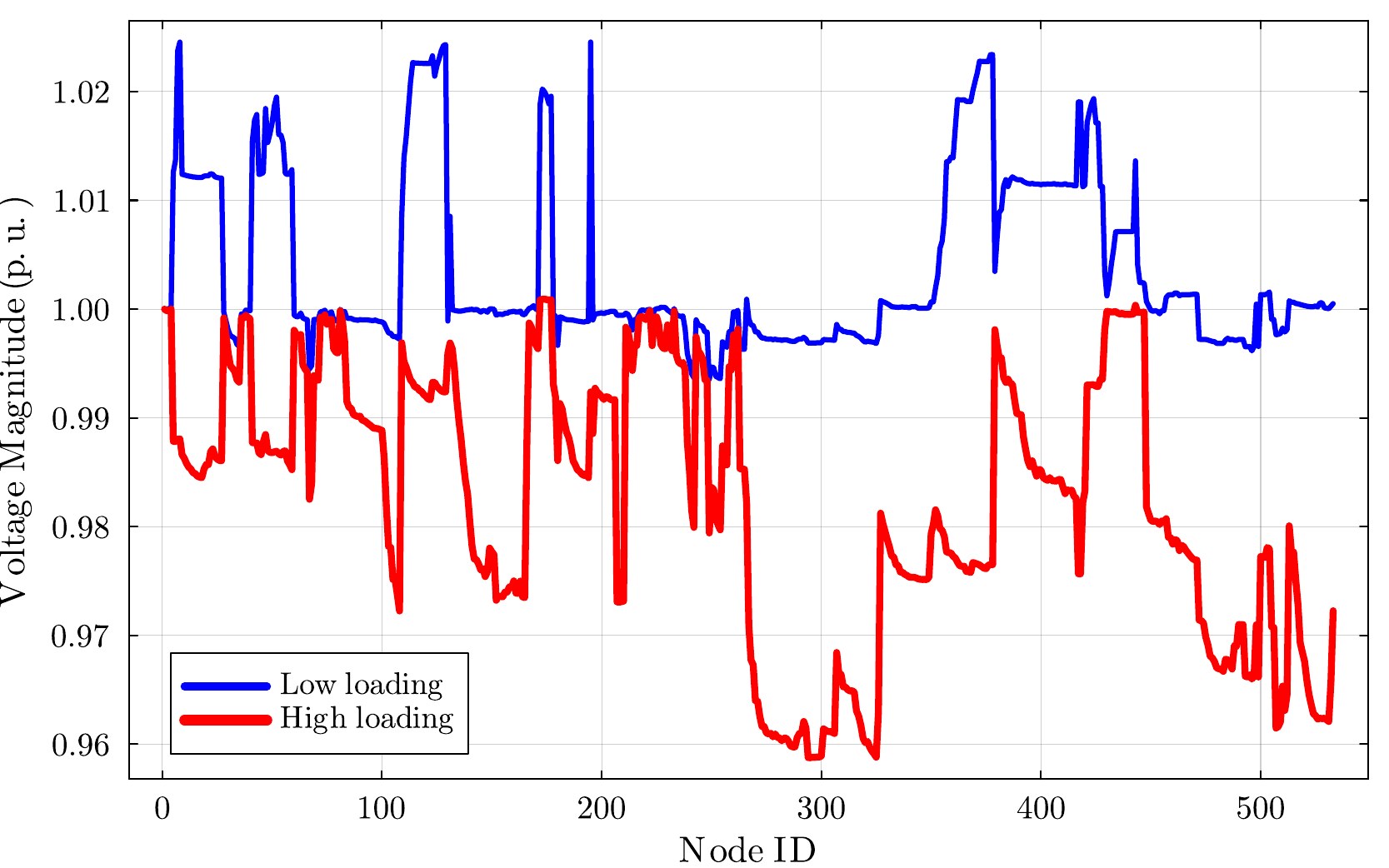}
    \caption{Voltage profile across 533-bus radial network. Low loading represents the net load of -5 MW and high loading represents the net load of 45 MW.}
    \label{fig: v-full}
\end{figure}
We applied the proposed model to this network and could find different reduction levels. The optimization parameter has been set to $\alpha=10/n$, and $q=1$. The effect of changing the voltage magnitude limit, $\Bar{E}$, is summarized in Table \ref{tab: E-dist}, where LL and HL denote low and high loading conditions, respectively.
\begin{table}[!t]
    \centering
    \caption{Comparison of high and low loading voltage errors for different reduction levels.}
    \label{tab: E-dist}
    \begin{adjustbox}{max width=1\textwidth}
    \begin{tabular}{r|cccccccc}
        \toprule
        $\Bar{E}$ & Red. & \multicolumn{2}{c}{Max ($m$p.u.)} & \multicolumn{2}{c}{Mean ($m$p.u.)} & \multicolumn{2}{c}{Median ($m$p.u.)} \\
        ($m$p.u.) & (\%) & LL & HL & LL & HL & LL & HL \\
        \midrule
        1.0   & 69      & 1.0 & 1.0     & 0.1 & 0.3     & 0.0 & 0.2 \\
        2.5   & 85      & 2.4 & 2.5     & 0.3 & 0.7     & 0.1 & 0.5 \\
        5.0     & 92      & 3.5 & 4.8     & 0.5 & 1.5     & 0.3 & 1.2 \\
        7.5   & 96      & 7.2 & 7.4     & 1.8 & 2.6     & 0.7 & 2.3 \\
        10    & 97      & 9.6 & 9.8     & 1.4 & 2.5     & 0.5 & 1.9 \\
        \bottomrule
    \end{tabular}
    \end{adjustbox}
\end{table}

As shown in Table~\ref{tab: E-dist}, applying the Opti-KRON formulation reduces the 533-bus distribution feeder by up to 97\% with acceptable maximum voltage magnitude errors (i.e., 0.01 p.u. or 10 $m$p.u.). However, these reduced networks are not necessarily radial.
Therefore, we apply our radialization approach to transform the Kron-reduced networks into radial ones. Table~\ref{tab:radialization_results} summarizes how many maximal cliques are present in each reduced network from Table~\ref{tab: E-dist} and shows the number of critical nodes that must be kept to restore a radial topology.

\begin{table}[!t]
    \centering
    \caption{Radialization Results for the 533-Bus Network}
    \label{tab:radialization_results}
    \begin{tabular}{c c c c}
        \toprule
        {Initial reduction} & Maximal cliques & Critical nodes & Final reduction \\
        \midrule
        69\% & 13 & 17 & 66\% \\
        85\% & 8  & 15 & 83\% \\
        92\% & 4  & 9  & 90\% \\
        96\% & 2  & 6  & 95\% \\
        97\% & 2  & 4  & 96\% \\
        \bottomrule
    \end{tabular}
\end{table}

As shown in Table~\ref{tab:radialization_results}, increased reduction levels yield fewer maximal cliques. Consequently, the number of critical nodes needed for radialization decreases. Overall, the cost of preserving radial topology—i.e., losing a small fraction in reduction—is relatively small.
The full graph of the network is demonstrated in Fig. \ref{fig: full-graph}, while super-nodes for the case of 85\% reduced network and their clusters have the same colors. 
\begin{figure}[!t]
    \centering
    \includegraphics[width=0.95\linewidth]{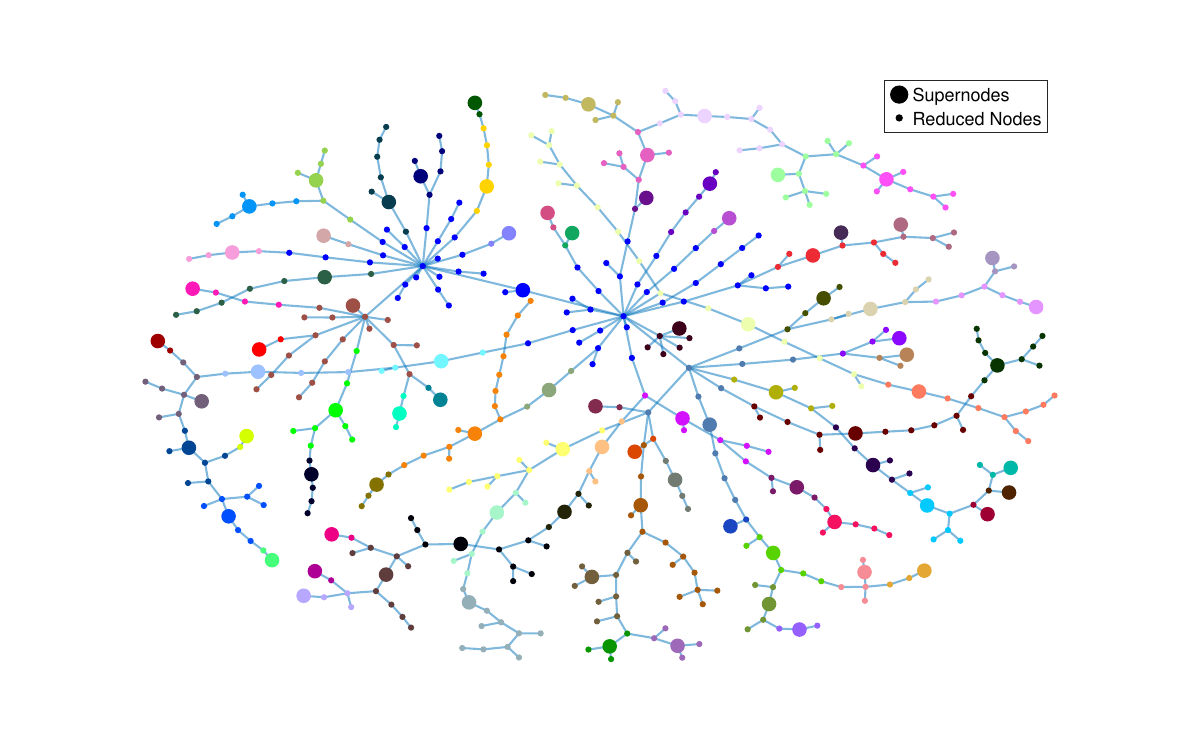}
    \caption{Graph visualization of the 533-bus full network. Super-nodes and their clusters are shown for 85\% reduced network. Nodes with the same color belong to the same cluster.}
    \label{fig: full-graph}
\end{figure}

To compare the accuracy and effectiveness of the reduced network on problems beyond basic power flow, we develop an optimization model to determine the reactive power setpoints for inverters located at distributed generation buses. The objective of this problem is to minimize the deviation of bus voltages from the nominal 1 p.u. value. The detailed mathematical formulation of this problem is presented in Appendix~\ref{appendix: Q-adjustment}. We applied this optimization to both the original and reduced versions of the 533-bus network under 100 loading scenarios. 
To generate synthetic load scenarios, we applied singular value decomposition~\cite{horn2012matrix} to historical load data from South Alburgh, Vermont, and fitted a Gaussian mixture model~\cite{bishop2006pattern} to model spatial and temporal correlations. We then generated scenarios by sampling from the mixture model and scaled them to match the load limits of the 533-bus system.
The voltage solution from the full network was used as the ground truth. After solving the optimization on the reduced network, the obtained reactive power setpoints were applied to a full network power flow, and the resulting voltages were compared to the ground truth. The maximum and mean voltage magnitude errors across all buses and scenarios are reported in Table~\ref{table: Q}.
\begin{table}[!t]
\centering
\caption{Computation time and voltage error for inverter dispatch optimization problem with different reduction levels.}
\label{table: Q}
    \begin{tabular}{c cc cc}
        \toprule
        \multirow{2}{*}{Reduction} & \multicolumn{2}{c}{Time (s)} & \multicolumn{2}{c}{Voltage Error ($m$p.u.)} \\
        & Meshed   & Radialized        & Maximum    & Mean  \\
        \midrule
        0\%~\rm (Full)& -     & 18.40   & 0.00      & 0.00    \\
        69\%         & 10.09 & 8.46    & 2.87      & 0.04    \\
        85\%         & 6.11  & 4.88    & 3.47      & 0.14    \\
        92\%         & 4.30  & 3.40    & 3.54      & 0.17    \\
        96\%         & 2.89  & 2.59    & 15.92     & 0.55    \\
        97\%         & 2.71  & 2.49    & 18.62     & 0.62    \\
        \bottomrule
    \end{tabular}
\end{table}

Table~\ref{table: Q} demonstrates that the proposed network reduction method significantly decreases computational time. While the radialized networks have more nodes, they introduce sparsity in the optimization problem, which speeds up solve times. Even for a 92\% reduction, the worst case voltage error remains below 0.004 p.u. (i.e., 4 $m$p.u.) and the resulting OPF problem was solved up to 5.4 times faster. 
\subsection{3499-bus South Alburgh feeder}

The South Alburgh feeder is a realistic utility distribution system from Vermont spanning voltage levels from 120~V to 26.6~kV. The original feeder is a three-phase network with 8387 nodes and 130 distributed generation units. In this work, we consider its phase-A subsystem, which yields a 3499-node feeder. From a week-long April dataset of 168 hourly loading scenarios, ranging from 0.73 to 1.2~MW, we selected three representative scenarios corresponding to low, medium, and high loading conditions for the reduction process.
We first applied the first stage reduction to reduce zero-injection nodes, using a voltage error threshold of $\bar{E} = 0.001$. We then applied radialization to increase network sparsity for the next stage. This resulted in an initial reduced network with 759 nodes, or a 78\% reduction, and an average voltage error of only 0.0002 p.u. In the second stage, we applied Opti-KRON to the initial reduced network. The optimization parameter has been set to $\alpha=20/n$, and $q=1$. Finally, we radialized the reduced networks.
Table~\ref{table: south-alburgh 1} represents the reduction result for this network, where ``\circled{1} vs. \circled{2}'' denotes the maximum absolute error relative to the Stage~1 reduced network, whereas ``Full vs. \circled{2}'' denotes the error relative to the full network.
\begin{table}[!t]
\centering
\caption{Stage \circled{2} reduction results on the South Alburgh feeder compared with Stage~\circled{1} and Full networks.}
\label{table: south-alburgh 1}
\begin{tabular}{c |cc c c}
\toprule
$\bar{E}$ & \multicolumn{2}{c}{Max Abs. Error ($m$p.u.)} & \multirow{2}{*}{Reduction} & \multicolumn{1}{c}{Opti-KRON} \\
    ($m$p.u.) & \circled{1} vs. \circled{2}    & Full vs. \circled{2} &  & Solve time (s)\\
\midrule
1.0  & 0.99     & 1.8   & 91\%       & 416.5 \\
2.5  & 2.49     & 3.1   & 93\%       & 481.6  \\
5.0  & 4.99     & 5.4   &  97\%       & 509.9  \\
7.5  & 7.78     & 7.8   &  98\%       &  531.5  \\
10   & 10.06    & 10.7  &  99\%       & 618.6  \\
\bottomrule
\end{tabular}
\end{table}
Since $\bar{E}$ restricts the voltage error of final reduced networks with respect to the initial reduced network from the first stage, maximum errors are slightly above $\bar{E}$.
 Additionally, the aggregated solve time is summarized in Table~\ref{table: south-alburgh 1}, which corresponds to the total time Gurobi spent solving the MILP across all iterations of Algorithm~\ref{alg: NR}. The average solve time per iteration for this network is approximately $0.75$ seconds, which demonstrates that the proposed formulation remains computationally tractable for large-scale networks.
To evaluate the robustness of the reduced network obtained using the three representative scenarios, we tested its performance across all 168 loading conditions from the original dataset. We solved a power flow on the 94\% reduced network for all the scenarios. Fig.~\ref{fig: histogram} shows the distribution of maximum voltage errors for this power flow problem.
\begin{figure}
    \centering
    \includegraphics[width=0.48\textwidth]{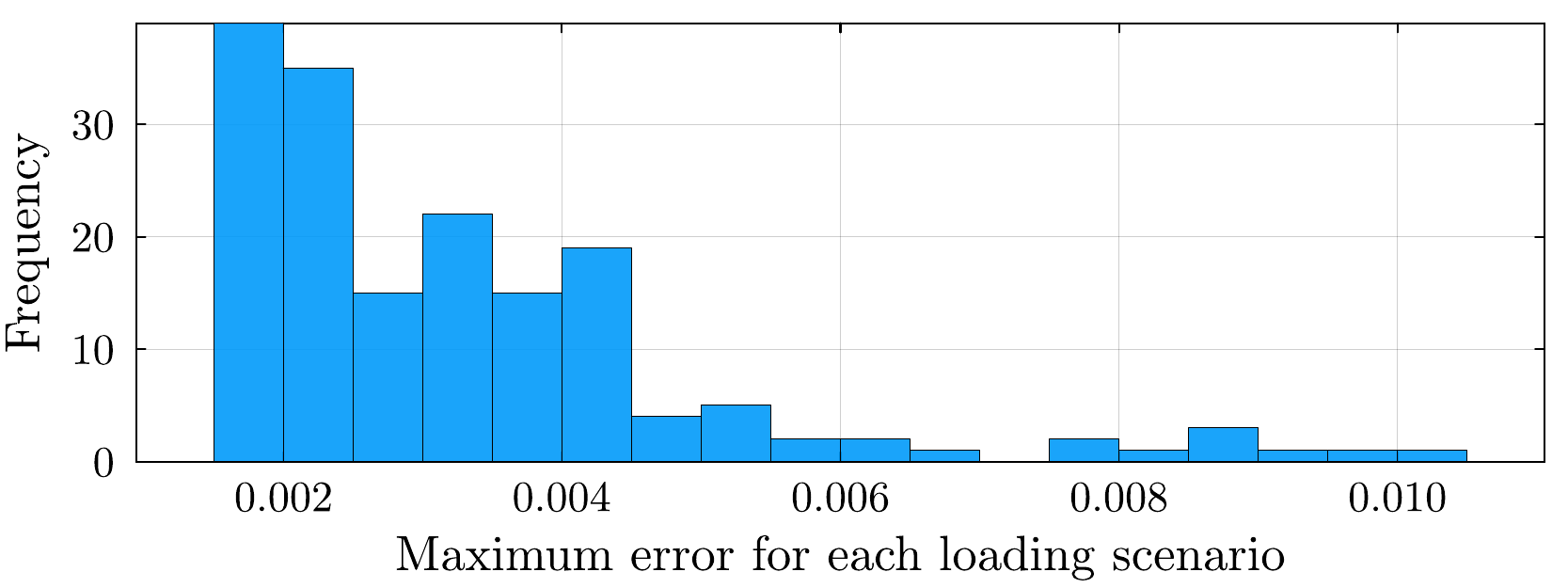}
    \caption{Histogram of maximum voltage errors across 168~load scenarios on the 94\% reduced South Alburgh feeder. Among the 168 scenarios, only 19 exhibit a maximum voltage error $\ge0.005$ p.u., each involving some combination of just 19 distinct nodes out of 3499.}
    \label{fig: histogram}
\end{figure}
Figure~\ref{fig: histogram} demonstrates that in over 89\% of the scenarios, the maximum voltage error remains below $0.005$ p.u. Across all scenarios, only 19 nodes exceed this threshold. Moreover, no scenario includes more than 2 such nodes among the 3499 buses.

To better visualize the trade-off between network reduction and accuracy, Fig.~\ref{fig:pareto} presents the Pareto front obtained from the reduction results for South Alburgh feeder. 
\begin{figure}[!t]
    \centering
    \includegraphics[width=0.48\textwidth]{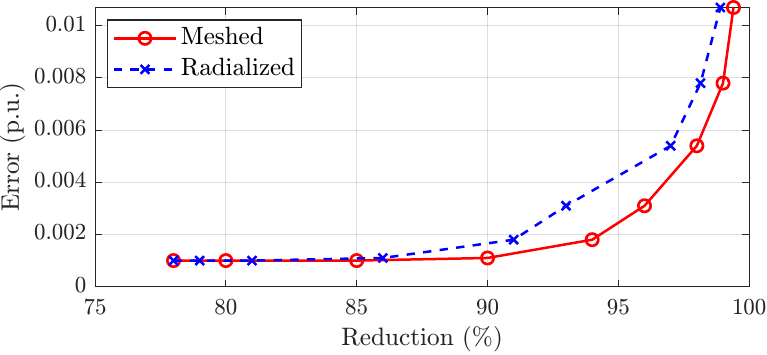}
    \caption{Maximum voltage magnitude error versus network reduction level for the South Alburgh feeder. The errors shown are from power flow solutions on the full and reduced networks.}
    \label{fig:pareto}
\end{figure}
The results show that both meshed and radialized reductions achieve substantial reductions even under tight voltage error limits, with meshed reductions about 2–3\% higher reduction levels. 

We applied the inverter dispatch optimal power flow problem (from Appendix~\ref{appendix: Q-adjustment}) to the South Alburgh feeder using 168 loading scenarios based on hourly load profiles over a week. Figure~\ref{fig: sa_opf} shows the maximum and mean voltage magnitude errors, along with OPF solve times, across different reduction levels. The solve times for pre-radialized reduced networks are not strictly decreasing with higher reductions. The reason is that smaller networks may still be more computationally intensive if they are more densely connected. However, this trend is not observed for radialized networks, where the solve time monotonically decreases with higher reductions.

\begin{figure}[!t]
    \centering
    \includegraphics[width=0.45\textwidth]{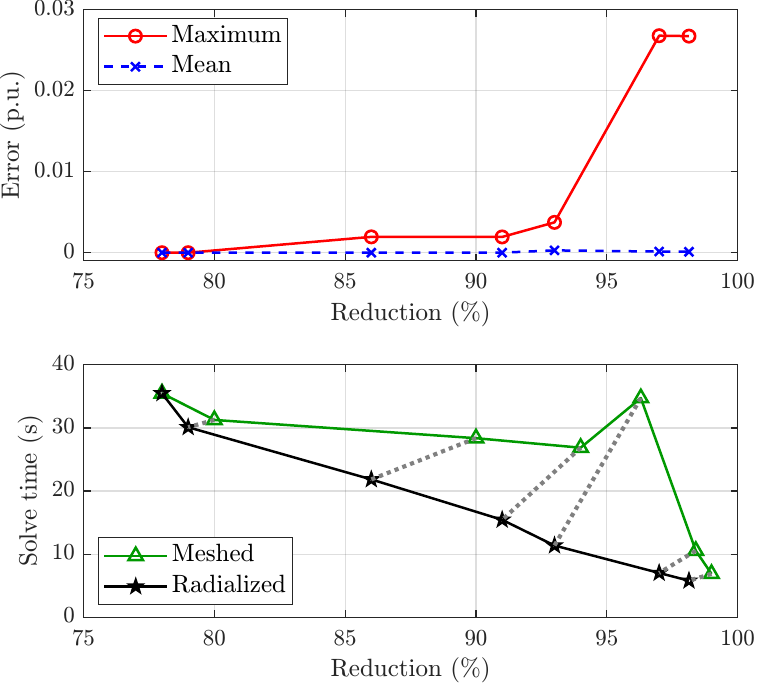}
    \caption{Voltage magnitude error and solver time for the inverter dispatch optimization problem on the reduced South Alburgh feeder under 168 loading scenarios. The voltage error plot reflects reduction levels based on the radialized networks. The 94\% reduced radialized network achieves a 14x speed-up (vs. 157.4 s for the full network) with maximum error below 0.002 p.u. While maximum error increases at high reduction levels, the mean error stays near zero, which indicates that the deviation for most nodes is negligible. }
    \label{fig: sa_opf}
\end{figure}

\subsection{Benchmarking, discussion, and limitations}
To the best of our knowledge, few established distribution network reduction methods jointly optimize node selection and preserve voltage profiles. Accordingly, the benchmarks considered here provide complementary comparisons for the two main components of the reduction problem: node selection and equivalent network construction.
We first compare Opti-KRON with~\cite{Mads_energies}, which uses electrical distance and K-means clustering followed by Kron reduction. While both methods combine clustering and reduction,~\cite{Mads_energies} relies on a similarity metric, whereas Opti-KRON directly incorporates voltage accuracy. This difference helps explain why Opti-KRON achieves lower voltage deviations on the South Alburgh network, as shown in Table~\ref{tab:IR_comparison}.

Second, we compare Opti-KRON with inversion reduction~\cite{Zachary_inversion_reduction_2018,Zachary_inversion_reduction_2019,Zachary_inversion_reduction_2021}, which preserves voltages at a set of kept nodes.
This approach does not determine the representative nodes or estimate voltages at reduced nodes.
Hence, for comparison, we use the super-nodes and assignments obtained by Opti-KRON.
Table~\ref{tab:IR_comparison} compares the resulting voltage errors.
This comparison isolates model construction error, not node selection performance.
Inversion reduction distributes reduced node injections among super-nodes, which can improve the accuracy of voltage at super-nodes.
In contrast, Opti-KRON assigns each reduced node to a representative super-node through a binary assignment matrix. This explicit assignment is useful when reduced nodes host loads, generators, or controllable devices, since each device remains associated with a known location in the reduced model rather than being implicitly redistributed. Consequently, inversion reduction is preferable when a salient node set is known and voltage accuracy is required only at those nodes, whereas Opti-KRON is preferable when the kept nodes are not known a priori or when reduced node voltages and device locations must remain interpretable for downstream optimization studies.
\begin{table}[t!]
\caption{Maximum Voltage Error Comparison.}
\label{tab:IR_comparison}
\centering
\begin{tabular}{c c|c c c}
\toprule
\multirow{2}{*}{Node type} &
\multirow{2}{*}{Reduction} &
\multicolumn{3}{c}{Max Error ($m$p.u.)} \\
\cmidrule(lr){3-5}
& & \cite{Mads_energies} & \cite{Zachary_inversion_reduction_2019} & Opti-KRON \\
\midrule
\multirow{3}{*}{Super-nodes}
& 80.0\% & 12.6 & \textbf{0.0} & \textbf{0.0} \\
& 90.0\% & 14.8 & \textbf{0.04} & 0.4 \\
& 95.0\% & 13.8 & \textbf{0.24} & 1.5 \\
\midrule
\multirow{3}{*}{Reduced nodes}
& 80.0\% & 22.5 & \textbf{1.0} & \textbf{1.0} \\
& 90.0\% & 22.6 & 1.3 & \textbf{1.1} \\
& 95.0\% & 26.9 & 3.9 & \textbf{3.1} \\
\midrule
\multicolumn{2}{c|}{Preserves loads} 
& {\color{ForestGreen}\checkmark} 
& \ding{55} 
& {\color{ForestGreen}\checkmark} \\ 
\bottomrule
\end{tabular}
\end{table}

Overall, the results highlight several practical implications of the proposed framework. Radialization is not only a structural requirement but also a computational advantage: although radialization reintroduces some nodes, the resulting sparsity can reduce the solution time of downstream voltage control problems relative to denser meshed Kron-reduced equivalents. 
The benchmark comparisons further demonstrate the value of optimization-based node selection over heuristic clustering.
The binary assignment matrix also provides an interpretable aggregation map by linking each reduced node to a specific super-node; controllable devices therefore remain associated with known locations in the reduced model rather than being implicitly redistributed.

The South Alburgh results suggest that reductions obtained from a small set of representative scenarios can generalize well across unseen loading conditions. Formal guarantees outside the training scenario set, however, remain an open problem. Adding more scenarios may improve robustness, but it also enlarges the optimization problem and can introduce additional scalability challenges. Similarly, the iterative cutting-plane strategy improved tractability and achieved accurate reductions in our tests, but it does not certify global optimality over all possible simultaneous node reductions.

\section{Conclusion} \label{sec: end}

This paper introduces a two-stage Kron-based network reduction approach that can be applied to distribution networks.
The proposed model performs reliably across diverse loading conditions and demonstrates robustness without significant voltage profile deviations. On a 533-bus radial distribution network, the method achieved an 85\% reduction and maintained voltage errors within acceptable limits. Similarly, the proposed two-stage reduction applied to a realistic 3499-bus Vermont feeder significantly reduced network size and ensured acceptable voltage accuracy.
Additionally, we introduced a radialization step that recovers a radial structure after Kron reduction. Our results showed that it preserves network structure and improves computational speed. 

The formulation developed in this paper is limited to balanced networks; an extension to unbalanced three-phase distribution feeders has been presented in~\cite{optikron3}. While the reduced networks satisfy the prescribed voltage error threshold on the training scenarios, stronger guarantees for unseen operating conditions remain to be established for future work. In addition, we are interested in studying how OPF solutions, distributed energy resource controller design, and hosting capacity allocations based on the reduced networks can be lifted to the full network and the associated optimality and feasibility gaps that result from such approaches.
If the gaps are practical, the Opti-KRON could potentially scale up distribution grid computations significantly across different domains.

\appendix
\subsection{Rectangular Reformulation of the Optimization Problem} \label{appendix: rec-decomp}
While the Opti-KRON formulation in~\eqref{opt_problem_final} is expressed using complex variables, it is reformulated in real rectangular coordinates so that standard optimization solvers can solve it. To do so, we first substitute~\eqref{eq: powerflow} by the following equations:
\begin{align}
    &\Re(\mathbf{Y}) V_l^{\text{Re}} - \Im(\mathbf{Y}) V_l^{\text{Im}} = A \Re(\hat{I}) \quad \forall l \in \mathcal{L}\label{eq: pf_real} \\
    &\Re(\mathbf{Y}) V_l^{\text{Im}} + \Im(\mathbf{Y}) V_l^{\text{Re}} = A \Im(\hat{I}) \quad \forall l \in \mathcal{L}  \label{eq: pf_imag},
\end{align}
where $V_l^{\text{Re}}$ and $V_l^{\text{Im}}$ are real and imaginary parts of $V_l$. Accordingly, we adjust~\eqref{eq: Big-M} for  $\forall i,j \in \mathcal{V}$ and $ \forall l \in \mathcal{L}$ as 
\begin{subequations} \label{eq: Big-M-dec}
    \begin{align}
    &W_{l,i,j}^{\text{Re}}  \leq (1-A_{i,j}) M_{l,i}^{\text{Re}} + V_{l,i}^{\text{Re}}\\
    &W_{l,i,j}^{\text{Re}}  \geq (A_{i,j}-1) M_{l,i}^{\text{Re}} + V_{l,i}^{\text{Re}} \\
    &W_{l,i,j}^{\text{Re}}  \leq M_{l,i}^{\text{Re}} A_{i,j}\\
    &W_{l,i,j}^{\text{Re}}  \geq-M_{l,i}^{\text{Re}} A_{i,j}. 
    \end{align}
\end{subequations}
 $W^{\text{Im}}$ is  similarly defined. Similarly,
\eqref{eq: error-opt_lin} is decomposed to
\begin{subequations}\label{eq: err_mat_dec}
\begin{align} 
    &E_{l}^{\text{Re}} = A {\rm diag} \{ \Re \{ \hat{V_l} \} \} - W_l^{\text{Re}} \\
    &E_{l}^{\text{Im}} = A {\rm diag}\{ \Im \{ \hat{V_l}\} \} - W_l^{\text{Im}} .
\end{align}
\end{subequations}
 $E_{l}^{\text{Re}}$ and $E_{l}^{\text{Im}}$ represent real and imaginary parts of the error matrix. Additionally,~\eqref{eq: errordef} transforms into
\begin{subequations} \label{eq: errordef-dec}
\begin{align} 
    &e_{l}^{\text{Re}} =  A^{\rm T}  V_{l}^{\text{Re}} - \Re \{\hat{V}_l\}\\
    &e_{l}^{\text{Im}} =  A^{\rm T} V_{l}^{\text{Im}} - \Im \{\hat{V}_l\}.
\end{align}
\end{subequations}
However,~\eqref{eq: v_abs_lin} is already decomposed into rectangular form.
Thus, Opti-KRON can be stated as:
\begin{subequations} \label{opt_dec}
\begin{align}
    \min_{\mathbb{E}} \quad &\sum_{l \in \mathcal{L}} \sum_{i \in \mathcal{V}} 
    \left( \left\Vert {\vec e_i^{\,T}} E_{l}^{\text{Re}} \right\Vert_{\infty} + \left\Vert {\vec e_i^{\,T}} E_{l}^{\text{Im}} \right\Vert_{\infty} \right) \notag \\
    & - \alpha \sum_{i \in \mathcal{V}} ({A^{\rm prev}_{i,i}} - A_{i,i})  \label{eq: obj-decomposed}\\
    \eqref{eq: assign_limit1}, \eqref{eq: assign_limit2}, &\eqref{eq :binary}-\eqref{eq: A_Omega},\eqref{eq: v_abs_lin},\eqref{eq: CP},\eqref{eq: pf_real}-\eqref{eq: errordef-dec} ,
\end{align} 
\end{subequations}
where $\mathbb{E} = \{A, \Omega, V^{\text{Re}}, V^{\text{Im}}, E^{\text{Re}}, E^{\text{Im}}, e^{\text{Re}}, e^{\text{Im}}, W^{\text{Re}}, W^{\text{Im}} \}$.

\subsection{Bounds on Big~M} \label{appendix: big-m} 
In this subsection, we derive tight yet feasible bounds on the Big~M constants in~\eqref{eq: Big-M-dec} to 
reduce solve time of~\eqref{opt_dec}. At each iteration $k$, we propose setting the Big~M constants to
\begin{subequations}
\begin{align}\label{eq_bigM_apdx}
M_{l,i}^{\text{Re}}(k) &:= \text{MICE}_{l,i}^{k-1} + \alpha q + \Re(\hat{V}_{l,i}), \\
M_{l,i}^{\text{Im}}(k) &:= \text{MICE}_{l,i}^{k-1} + \alpha q + \Im(\hat{V}_{l,i}), \label{eq_bigM_apdx2}
\end{align}
\end{subequations}
where $\text{MICE}_{l,i}^{k-1} := \left\Vert {\vec e_i^{\,T}} E_{l}^{\text{Re}} \right\Vert_{\infty} + \left\Vert {\vec e_i^{\,T}} E_{l}^{\text{Im}}\right\Vert_{\infty}$ is computed using the optimal solution from iteration $k-1$. $E_{l}^{\text{Re}}$ and $E_{l}^{\text{Im}}$ are introduced in~\eqref{eq: err_mat_dec}.
We now prove that the proposed bound for ${M_{l,i}^{\text{Re}}}^k$ is feasible and tight relative to the value of MICE. To preserve feasibility of \eqref{eq: Big-M-dec}, we require $M_{l,i}^{\mathrm{Re}} \ge V_{l,i}^{\mathrm{Re}}=\Re(\hat{V}_{l,i}) + e_{l,i}^{\mathrm{Re}}$.
Accordingly, we need to find an upper bound on $e_{l,i}^{\mathrm{Re}}$ to obtain a feasible value for $M_{l,i}^{\mathrm{Re}}$.
We can rewrite the objective function at iteration $k$ as
\begin{align}
    &O^k = \sum_{\text{l} \in \mathcal{L}} \sum_{\text{i} \in \mathcal{V}} \text{MICE}_{l,i}^k  
    - \alpha \sum_{i} ({A^{\rm prev}_{i,i}} - A_{i,i}^k).
\end{align}
The term $-\alpha\sum_{i}(A^{\text{prev.}}_{i,i}-A^{\,k}_{i,i})$ incentivizes further node reduction: it decreases the objective by $\alpha$ for each node reduced in iteration $k$ and is zero otherwise. When this decrease outweighs the change in the first term between iterations $k-1$ and $k$, more nodes are reduced. 
Accordingly,
{\small
\begin{align}
   \alpha \sum_{i\in\mathcal{V}}
      \bigl(A^{\text{prev.}}_{i,i}-A^{\,k}_{i,i}\bigr)
    &\ge
   \sum_{l\in\mathcal{L}}\sum_{i\in\mathcal{V}}
      \bigl(\mathrm{MICE}_{l,i}^k-\mathrm{MICE}_{l,i}^{\,k-1}\bigr).
   \label{eq: MICE-diff-bound} \\
    & \ge  \text{MICE}_{l,i}^k - \text{MICE}_{l,i}^{k-1}.
\end{align}}
The cutting-plane constraint~\eqref{eq: CP} limits the number of nodes that can be reduced in any iteration to at most $q$. Therefore,
\begin{align}
    \alpha q\ge  \text{MICE}_{l,i}^k - \text{MICE}_{l,i}^{k-1},
\end{align}
which yields
\begin{align}
    \alpha q  +  \text{MICE}_{l,i}^{k-1} \ge  \left\Vert {\vec e_i^{\,T}} E_{l}^{\text{Re}} \right\Vert_{\infty} + \left\Vert {\vec e_i^{\,T}} E_{l}^{\text{Im}} \right\Vert_{\infty}.
\end{align}
Since $e^{\rm Re}_{l,i} \le \left\Vert {\vec e_i^{\,T}} E_{l}^{\text{Re}} \right\Vert_{\infty}$, we have $\alpha q  +  \text{MICE}_{l,i}^{k-1} \geq e^{\rm Re}_{l,i}$, which provides an upper bound on $e_{l,i}^{\mathrm{Re}}$. The derivation for ${M_{l,i}^{\text{Im}}}^k$ follows analogously. Thus, the Big~M constants in~\eqref{eq_bigM_apdx} and~\eqref{eq_bigM_apdx2} are guaranteed to be feasible.

\subsection{Reactive Power Adjustment Optimization Problem} \label{appendix: Q-adjustment}
This appendix presents the mathematical formulation of the reactive power optimization problem used to validate the reduced networks. The objective is to adjust inverter reactive power setpoints to minimize voltage magnitude deviations with AC power flow and operational constraints.
\begin{subequations}
\begin{align}
&\min_{\substack{|V|,\ \theta, \ q^{\text{inv}},\ p^{\rm g}},\ q^{\rm g} }\quad 
 \sum_{i \in \mathcal{N}} \left| |V_{i}| - 1 \right| \label{eq:Q-obj} \\
\text{s.t.} \quad 
& P_{i} + p^{\rm g}_{i | i= \rm slack} = \sum_{k \in \mathcal{N}} |V_{i}||V_{k}| G_{i,k} \cos(\theta_{i} - \theta_{k}) \notag\\
&+|V_{i}||V_{k}| B_{i,k} \sin(\theta_{i} - \theta_{k}) \quad \forall i \label{eq:Q-PF-real} \\
& Q_{i} + q^{\rm g}_{i| i= \rm slack} + q_{i}^{\text{inv}} = \sum_{k \in \mathcal{N}} |V_{i}||V_{k}| G_{i,k} \sin(\theta_{i} - \theta_{k}) \notag \\
&- |V_{i}||V_{k}| B_{i,k} \cos(\theta_{i} - \theta_{k}) \quad \forall i\label{eq:Q-PF-reactive} \\
& (q_{i}^{\text{inv}})^2 \leq S_{\max,i}^2 - (p_{i}^{\text{inv}})^2 \quad \forall i \in \mathcal{S} \label{eq:Q-inv-capacity} \\
& V_{\min} \leq |V_{i}| \leq V_{\max} \quad \forall i. \label{eq:Q-voltage-bounds} 
\end{align}  
\end{subequations}
Equation~\eqref{eq:Q-obj} defines the objective of minimizing the sum of voltage magnitude deviations from 1 p.u. across all buses.
Equations~\eqref{eq:Q-PF-real} and~\eqref{eq:Q-PF-reactive} represent the nonlinear AC power flow constraints for real and reactive power, including inverter contributions. Parameters $P$ and $Q$ are the injections of load, while $p^{\rm g}$ and $q^{\rm g}$ are variables that represent slack bus active and reactive outputs.
Constraint~\eqref{eq:Q-inv-capacity} ensures that each inverter operates within its apparent power limit.
Constraint~\eqref{eq:Q-voltage-bounds} enforces voltage magnitude bounds.

\printbibliography

@INPROCEEDINGS{optiKRON2022,
  author={Chevalier, Samuel and Almassalkhi, Mads R.},
  booktitle={2022 IEEE 61st Conference on Decision and Control (CDC)}, 
  title={Towards Optimal {Kron-based} Reduction Of Networks {(Opti-KRON)} for the Electric Power Grid}, 
  year={2022},
  volume={},
  number={},
  pages={5713-5718},
  doi={10.1109/CDC51059.2022.9992730}}

@ARTICLE{PTDF_2012,
  author={Oh, HyungSeon},
  journal={IEEE Transactions on Power Systems}, 
  title={Aggregation of Buses for a Network Reduction}, 
  year={2012},
  volume={27},
  number={2},
  pages={705-712},
  doi={10.1109/TPWRS.2011.2176758}}

@ARTICLE{TNR_multicut_kron_2018,
  author={Ploussard, Quentin and Olmos, Luis and Ramos, Andrés},
  journal={IEEE Transactions on Power Systems}, 
  title={An Efficient Network Reduction Method for Transmission Expansion Planning Using Multicut Problem and {Kron} Reduction}, 
  year={2018},
  volume={33},
  number={6},
  pages={6120-6130},
  doi={10.1109/TPWRS.2018.2842301}}

@ARTICLE{PTDF_2015,
  author={Shi, Di and Tylavsky, Daniel J.},
  journal={IEEE Transactions on Power Systems}, 
  title={A Novel Bus-Aggregation-Based Structure-Preserving Power System Equivalent}, 
  year={2015},
  volume={30},
  number={4},
  pages={1977-1986},
  doi={10.1109/TPWRS.2014.2359447}}

@ARTICLE{ward_1949,
  author={Ward, J. B.},
  journal={Electrical Engineering}, 
  title={Equivalent circuits for power-flow studies}, 
  year={1949},
  volume={68},
  number={9},
  pages={794-794},
  doi={10.1109/EE.1949.6444973}}

@ARTICLE{kron_2013,
  author={Dorfler, Florian and Bullo, Francesco},
  journal={IEEE Transactions on Circuits and Systems I: Regular Papers}, 
  title={Kron Reduction of Graphs With Applications to Electrical Networks}, 
  year={2013},
  volume={60},
  number={1},
  pages={150-163},
  doi={10.1109/TCSI.2012.2215780}}

@INPROCEEDINGS{Taheri_parameter_learning_2024,
  author={Taheri, Babak and Molzahn, Daniel K.},
  booktitle={2024 IEEE Texas Power and Energy Conference (TPEC)}, 
  title={{AC} Power Flow Informed Parameter Learning for {DC} Power Flow Network Equivalents}, 
  year={2024},
  volume={},
  number={},
  pages={1-6},
  doi={10.1109/TPEC60005.2024.10472173}}

@INPROCEEDINGS{jiang_zone_2023,
  author={Jiang, Yanda and Parvini, Zohreh and McCalley, James and Lhuillier, Nicolas and Despouys, Olivier and Figueroa-Acevedo, Armando and Okullo, James},
  booktitle={2023 North American Power Symposium (NAPS)}, 
  title={Network Reduction for Power System Planning: Zone Identification}, 
  year={2023},
  volume={},
  number={},
  pages={1-6},
  doi={10.1109/NAPS58826.2023.10318645}}

@ARTICLE{Ribeiro_DD_TNR_2024,
  author={Ribeiro, Raul and Street, Alexandre and Mancilla–David, Fernando and Angulo, Alejandro},
  journal={IEEE Transactions on Power Systems}, 
  title={Equivalent Reduced {DC} Network Models With Nonlinear Load Functions: A Data-Driven Approach}, 
  year={2024},
  volume={39},
  number={2},
  pages={3021-3032},
  doi={10.1109/TPWRS.2023.3274919}}

@ARTICLE{Zhu_Optimization_DC_NR_2018,
  author={Zhu, Yujia and Tylavsky, Daniel},
  journal={IEEE Transactions on Power Systems}, 
  title={An Optimization-Based {DC}-Network Reduction Method}, 
  year={2018},
  volume={33},
  number={3},
  pages={2509-2517},
  doi={10.1109/TPWRS.2017.2745492}}

@ARTICLE{Zachary_inversion_reduction_2021,
  author={Pecenak, Zachary K. and Haghi, Hamed Valizadeh and Li, Changfu and Reno, Matthew J. and Disfani, Vahid R. and Kleissl, Jan},
  journal={IEEE Transactions on Smart Grid}, 
  title={Aggregation of Voltage-Controlled Devices During Distribution Network Reduction}, 
  year={2021},
  volume={12},
  number={1},
  pages={33-42},
  doi={10.1109/TSG.2020.3011073}}

@ARTICLE{Zachary_inversion_reduction_2019,
  author={Pecenak, Zachary K. and Disfani, Vahid R. and Reno, Matthew J. and Kleissl, Jan},
  journal={IEEE Transactions on Power Systems}, 
  title={Inversion Reduction Method for Real and Complex Distribution Feeder Models}, 
  year={2019},
  volume={34},
  number={2},
  pages={1161-1170},
  doi={10.1109/TPWRS.2018.2872747}}

@ARTICLE{Zachary_inversion_reduction_2018,
  author={Pecenak, Zachary K. and Disfani, Vahid R. and Reno, Matthew J. and Kleissl, Jan},
  journal={IEEE Transactions on Power Systems}, 
  title={Multiphase Distribution Feeder Reduction}, 
  year={2018},
  volume={33},
  number={2},
  pages={1320-1328},
  doi={10.1109/TPWRS.2017.2726502}}

@misc{9124806,
  author       = {{Malmer, Gabriel and Thorin, Lovisa}},
  language     = {{eng}},
  note         = {{Student Paper}},
  series       = {{CODEN:LUTEDX/TEIE}},
  title        = {{Network reconfiguration for renewable generation maximization}},
  year         = {2023},
}

@ARTICLE{inverse_pf,
  author={Yuan, Ye and Low, Steven H. and Ardakanian, Omid and Tomlin, Claire J.},
  journal={IEEE Transactions on Control of Network Systems}, 
  title={Inverse Power Flow Problem}, 
  year={2023},
  volume={10},
  number={1},
  pages={261-273},
  doi={10.1109/TCNS.2022.3199084}}

@ARTICLE{dist_flow,
  author={Baran, M. and Wu, F.F.},
  journal={IEEE Transactions on Power Delivery}, 
  title={Optimal sizing of capacitors placed on a radial distribution system}, 
  year={1989},
  volume={4},
  number={1},
  pages={735-743},
  doi={10.1109/61.19266}}

@ARTICLE{BFS,
  author={Chang, G. W. and Chu, S. Y. and Wang, H. L.},
  journal={IEEE Transactions on Power Systems}, 
  title={An Improved Backward/Forward Sweep Load Flow Algorithm for Radial Distribution Systems}, 
  year={2007},
  volume={22},
  number={2},
  pages={882-884},
  doi={10.1109/TPWRS.2007.894848}}

@ARTICLE{Jabr_radial,
  author={Jabr, R.A.},
  journal={IEEE Transactions on Power Systems}, 
  title={Radial distribution load flow using conic programming}, 
  year={2006},
  volume={21},
  number={3},
  pages={1458-1459},
  doi={10.1109/TPWRS.2006.879234}}

@misc{benjamin,
      title={Congestion-Sensitive Grid Aggregation for {DC} Optimal Power Flow}, 
      author={Benjamin Stöckl and Yannick Werner and Sonja Wogrin},
      year={2025},
      eprint={2505.07545},
      archivePrefix={arXiv},
    url={https://arxiv.org/abs/2505.07545},
      %primaryClass={math.OC}, 
}

@Article{Mads_energies,
AUTHOR = {Almassalkhi, Mads and Brahma, Sarnaduti and Nazir, Nawaf and Ossareh, Hamid and Racherla, Pavan and Kundu, Soumya and Nandanoori, Sai Pushpak and Ramachandran, Thiagarajan and Singhal, Ankit and Gayme, Dennice and Ji, Chengda and Mallada, Enrique and Shen, Yue and You, Pengcheng and Anand, Dhananjay},
TITLE = {Hierarchical, Grid-Aware, and Economically Optimal Coordination of Distributed Energy Resources in Realistic Distribution Systems},
JOURNAL = {Energies},
VOLUME = {13},
YEAR = {2020},
NUMBER = {23},
ARTICLE-NUMBER = {6399},
URL = {https://www.mdpi.com/1996-1073/13/23/6399},
ISSN = {1996-1073},
DOI = {10.3390/en13236399}
}

@article{JuMP,
    author = {Miles Lubin and Oscar Dowson and Joaquim {Dias Garcia} and Joey Huchette and Beno{\^i}t Legat and Juan Pablo Vielma},
    title = {{JuMP} 1.0: {R}ecent improvements to a modeling language for mathematical optimization},
    journal = {Mathematical Programming Computation},
    year = {2023},
    doi = {10.1007/s12532-023-00239-3}
}

@book{diestel2018graph,
  title={Graph Theory},
  author={Diestel, R.},
  isbn={9783662575604},
  series={Graduate Texts in Mathematics},
  url={https://books.google.com/books?id=JQ1XtgEACAAJ},
  year={2018},
  publisher={Springer Berlin Heidelberg}
}

@ARTICLE{Kekatos_graph_id,
  author={Cavraro, Guido and Kekatos, Vassilis},
  journal={IEEE Control Systems Letters}, 
  title={Graph Algorithms for Topology Identification Using Power Grid Probing}, 
  year={2018},
  volume={2},
  number={4},
  pages={689-694},
  doi={10.1109/LCSYS.2018.2846801}}

@techreport{DOE_report,
  author       = {},
  title        = {Grid Modernization Strategy 2024},
  institution  = {U.S. Department of Energy},
  year         = {2024},
  type         = {Tech. Rep.},
  url          = {https://www.energy.gov/sites/default/files/2024-12/Grid%20Modernization%20Strategy%202024.pdf},
}

@book{horn2012matrix,
  title={Matrix analysis},
  author={Horn, Roger A and Johnson, Charles R},
  year={2012},
  publisher={Cambridge university press}
}

@book{bishop2006pattern,
  title={Pattern recognition and machine learning},
  author={Bishop, Christopher M and Nasrabadi, Nasser M},
  volume={4},
  number={4},
  year={2006},
  publisher={Springer}
}

@article{bigM,
  title={Cutting {Big M} down to size},
  author={Camm, Jeffrey D and Raturi, Amitabh S and Tsubakitani, Shigeru},
  journal={Interfaces},
  volume={20},
  number={5},
  pages={61--66},
  year={1990},
  publisher={INFORMS}
}

@INPROCEEDINGS{graine_2023_dnr,
  author={Graine, Aghyles and Karnib, Nour and Grolleau, Emmanuel and Bertout, Antoine and Gaubert, Jean-Paul and Larraillet, Didier},
  booktitle={2023 IEEE PES Innovative Smart Grid Technologies Europe (ISGT EUROPE)}, 
  title={A Network Reduction Method for the Distribution Network Reconfiguration Problem}, 
  year={2023},
  volume={},
  number={},
  pages={1-5},
  doi={10.1109/ISGTEUROPE56780.2023.10407566}}

@ARTICLE{todorovski_2024_dnr,
  author={Todorovski, Mirko},
  journal={IEEE Transactions on Power Systems}, 
  title={A Reduction Method for Radial Distribution Feeders: Ensuring Parity in Voltages and Losses}, 
  year={2024},
  volume={39},
  number={2},
  pages={4759-4762},
  doi={10.1109/TPWRS.2024.3351490}}

@ARTICLE{huang_2024_drl,
  author={Huang, Bin and Wang, Jianhui},
  journal={IEEE Transactions on Network Science and Engineering}, 
  title={Adaptive Static Equivalences for Active Distribution Networks With Massive Renewable Energy Integration: A Distributed Deep Reinforcement Learning Approach}, 
  year={2024},
  volume={11},
  number={6},
  pages={5463-5476},
  doi={10.1109/TNSE.2023.3272794}}

@ARTICLE{meyn_2025_ccsc,
  author={Baquedano-Aguilar, Mario D. and Meyn, Sean and Bretas, Arturo},
  journal={IEEE Open Access Journal of Power and Energy}, 
  title={Coherency-Constrained Spectral Clustering for Power Network Reduction}, 
  year={2025},
  volume={12},
  number={},
  pages={88-99},
  doi={10.1109/OAJPE.2025.3538619}}

@INPROCEEDINGS{huang_2024_pesgm,
  author={Huang, Bin and Li, Jiayong and Guo, Han and Wang, Jianhui},
  booktitle={2024 IEEE Power \& Energy Society General Meeting (PESGM)}, 
  title={Evidential Reasoning for Enhanced Node Selection in Power Network Reduction: a Complex Network Perspective}, 
  year={2024},
  volume={},
  number={},
  pages={1-5},
  doi={10.1109/PESGM51994.2024.10689175}}

@article{optikron3,
  title={Optimal {Kron-based} Reduction of Networks {(Opti-KRON)} for Three-phase Distribution Feeders},
  author={Mokhtari, Omid and Chevalier, Samuel and Almassalkhi, Mads},
  journal={arXiv preprint arXiv:2510.19608},
  year={2025}
}

@ARTICLE{Bolognani_2016,
  author={Bolognani, Saverio and Zampieri, Sandro},
  journal={IEEE Transactions on Power Systems}, 
  title={On the Existence and Linear Approximation of the Power Flow Solution in Power Distribution Networks}, 
  year={2016},
  volume={31},
  number={1},
  pages={163-172},
  doi={10.1109/TPWRS.2015.2395452}}
\end{document}